 \documentclass[final,3p,authoryear]{elsarticle}

\pdfoutput=1
\usepackage{graphicx}
\usepackage{caption}

\usepackage[utopia]{mathdesign}
\usepackage[OMLmathrm,OMLmathbf]{isomath} 
    \usepackage{lmodern}
\usepackage{amsmath}
\usepackage{amsfonts}
\usepackage{graphicx}
\usepackage{dsfont}
\usepackage{color}
\usepackage{tensor}
\usepackage[mathscr]{eucal}
\numberwithin{equation}{section}

\def \alphaII{{\mathbfit{\alpha}}^{\mathrm{II}} }
\def \kappaII{{\mathbfit{\kappa}}^{\mathrm{II}} }
\def \kappavec{{\mathbfit{\kappa}} }

\def \km{{\left(km\right)}}
\def \mm{{\left(m\right)}}
\def \R{{\mathbfit{R}} }
\def \Y{{\mathbfit{Y}} }

\def \Psibf{{\mathbfit{\Psi}} }
\def \Phibf{{\mathbfit{\Phi}} }

\def \q{{\mathbfit{q}} }
\def \Q{{\mathbfit{Q}} }

\def \L{{\mathbfit{L}} }
\def \H{{\mathbfit{H}} }
\def \J{{\mathbfit{J}} }
\def \l{{\mathbfit{l}} }
\def \V{{\mathbfit{V}} }
\def \v{{\mathbfit{v}} }
\def \e{{\mathbfit{h}} }
\def \i{{\mathrm{\iota}}}

\def \Vel{{\mathrm{V}}}
\def \Lie{{\mathscr{L}}}

\def \grad{{\nabla}}
\def \dex{{\operatorname{d}}}

\def \bpartial{{\boldsymbol{\partial}}}
\def \bparx{{\bpartial_1}}
\def \bpary{{\bpartial_2}}
\def \bparz{{\bpartial_3}}

\def \parphi{{\partial_{\varphi}}}
\def \b{{\mathbfit{b}}}
\def \g{{\mathbfit{g}}}

\def \lhat{{\hat{l}}}
\def \n{{\mathbfit{n}}}
\def \m{{\mathbfit{m}}}

\def \u{{\mathbfit{u}}}
\def \dV{{\ensuremath{dV}}}
\def \dVS{{\ensuremath{dU}}}
\def \dS{{\ensuremath{dS}}}
\def \kappabf{{\mathbfit{\kappa}}^{\mathrm{b}}}
\def \betapl{{\mathbfit{\beta}^{\mathrm{p}}}}
\def \thetabf{{\mathbfit{\vartheta}}}
\def \alphabf{{\mathbfit{\alpha}}}
\def \etabf{{\mathbfit{\eta}}}
\def \rhobf{{\mathbfit{\rho}}}
\def \totimes{{\widehat{\otimes}}}
\def \sotimes{{\otimes_\mathrm{s}}}

\newcommand  {\dx 	}[1]{\ensuremath{\dex x^{#1}} }
\newcommand  {\tr 	}{\ensuremath{\operatorname{tr}} }

\newcommand  {\divergence}{\nabla \cdot} 
\newcommand  {\curl}{\ensuremath{\nabla \times} }
\renewcommand  {\phi}{\varphi}
\usepackage{amsthm}
\newtheorem{proposition}{Proposition}
\newtheorem{definition}{Definition}

\journal{Journal of the Mechanics and Physics of Solids}

\begin{document}

\begin{frontmatter}



\title{Multipole expansion of continuum dislocations dynamics in terms of alignment tensors}

\author[label1]{Thomas Hochrainer}

\ead{hochrainer@mechanik.uni-bremen.de}
\address[label1]{BIME - Bremer Institut f\"ur Strukturmechanik und Produktionsanlagen, Universit\"at Bremen, IW3, Am Biologischen Garten 2, 28359 Bremen, Germany}

\begin{abstract}
Dislocation based modeling of plasticity is one of the central challenges at the crossover of materials science and continuum mechanics. Developing a continuum theory of dislocations requires the solution of two long standing problems: (\emph{i}) to find a faithful representation of dislocation kinematics with a reasonable number of variables and (\emph{ii}) to derive averaged descriptions of the dislocation dynamics (i.e. material laws) in terms of these variables.  In the current paper we solve the first problem. This is achieved through a multipole expansion of the dislocation density in terms of so-called alignment tensors containing the directional distribution of dislocation density and dislocation curvature. A hierarchy of evolution equations of these tensors is derived from a higher dimensional dislocation density theory. Low order closure approximations of this hierarchy lead to continuum dislocation dynamics models with only few internal variables. Perspectives for more refined theories and current challenges in dislocation density modeling are discussed. 
\end{abstract}


\end{frontmatter}


\section{Introduction} \label{Sec: Introduction}
About 60 years after the introduction of the dislocation density tensor introduced independently and largely simultaneously by \citet{kondo52, nye53, bilby_bs55} and \citet{kroener_r56}, a dislocation density based theory of plasticity is still not available. Of course, phenomenological hardening laws for polycrystals benefit a lot from dislocation theory \citep{kocks_m03, devincre_hk08} and composite models \'a la \citet{Mughrabi19831367} have come a long way in describing dislocation structures. But with regard to a plasticity theory based on controlled averaging of the behavior of an ensemble of dislocations the pessimistic summary presented by \citet{kroener01} in his last paper is essential still valid. Kr\"oner himself remained dissatisfied with the incomplete information contained in the dislocation density tensor his whole life. In the course of time he suggested several ways how to overcome this deficiency. One of them was the introduction of so-called correlation tensors \citep{kroener69} which generalizes the concept of higher order correlation functions from statistical mechanics to the tensorial description needed for dislocations. With this idea he was far ahead of his time but he had no chance of obtaining reasonable information on the pair correlation tensors back then. The concept has been introduced successfully in dislocation theory in a series of papers around the the last turn of century \citep{groma97, zaiser_mg01, groma_cz03}, but only for the non-tensorial case of straight parallel edge dislocations. In these works pair correlations were derived from a series of two-dimensional discrete dislocation simulations. Obtaining the pair correlation tensors for fully three-dimensional ensembles of curved dislocations has not yet been seriously attempted, aside from two preliminary studies by \citet{csikor_etal07} and \citet{deng_e07}. But while pair correlations are of obvious importance for developing an averaged theory of dislocations they do not contribute to completing the insufficient information on the dislocation state contained in the dislocation density tensor. This completion requires the definition of dislocation density measures which contain enough information to predict at least kinematically the development of plastic deformation and their own evolution. We call such theories \emph{continuum dislocation dynamics} (CDD) theories. In the following we give a brief overview of existing CDD theories and recapitulate the kinematic challenge in averaging dislocation systems.

The basic CDD theory is what we consider a pseudo-continuum theory of moving dislocations based solely on the Kröner-Nye-tensor $ \alphabf = \curl \mathbfit{\beta}^\mathrm{p} $, which derives as the \emph{curl} of the plastic distortion tensor $ \betapl $. This elementary CDD theory is valid if the spatial resolution is such that all dislocations are resolved individually or if dislocations are continuously distributed but all dislocations are \emph{geometrically necessary dislocations} (GND). In this case the Kröner-Nye-tensor contains all information on the dislocation state. Necessarily this means that the Kröner-Nye tensor needs to be known on a level where the dislocation line direction $ \l $ can be inferred from it. Only then the local dislocation velocity vector $ \v $ which is orthogonal to $ \l $ can be defined uniquely. This is the case if the Kröner-Nye tensor is known on each slip system $\varsigma$ separately, $ \alphabf_\varsigma $, such that the overall Kröner Nye-tensor may be obtained as \ $ \alphabf = \sum_\varsigma \alphabf_\varsigma $. On each slip system the Kröner-Nye-tensor is then decomposable into a tensor product of a dislocation density vector $ \kappavec_\varsigma $ and the Burgers vector $ \b_\varsigma $ as $ \alphabf_\varsigma = \kappavec_\varsigma \otimes \b_\varsigma $. The total dislocation density $ \rho_\varsigma $ per slip system is in this case given by the norm of the dislocation density vector, $ \rho_\varsigma =\kappa_\varsigma $, such that $ \alphabf_\varsigma = \rho_\varsigma \l_\varsigma \otimes \b_\varsigma $. Given the line direction $ \l_\varsigma $ the dislocation velocity vector $ \v_\varsigma $ may be derived from the current stress state by some material law involving the Peach-Koehler-force. Leaving aside the concrete form of the material law we obtain the rate of plastic distortion tensor as 
\begin{equation}
	\partial_t \mathbfit{\beta}^\mathrm{p}_\varsigma = \v_\varsigma \times \alphabf_\varsigma.
\end{equation}
In the case of only GND we consequently find the evolution equation as introduced for example by \citet{mura63}
\begin{equation} \label{Eq: mura}
	\partial_t \alphabf_\varsigma = \curl \left( \v_\varsigma \times \alphabf_\varsigma \right), 
\end{equation}
which translates due to the fixed Burgers vector directly into the form for the dislocation density vector,
\begin{equation} \label{Eq: evolution kappa intro}
	\partial_t \kappavec_\varsigma = \curl \left( \v_\varsigma \times \kappavec_\varsigma \right). 
\end{equation}
At the heart of the challenge in developing an averaged description of dislocations lies the fact that in averaging the right hand side of Eq.\ \eqref{Eq: evolution kappa intro} the cross product does not commute with the averaging \citep{acharya_r06}. That means one generally finds $ \overline{ \v_\varsigma \times \kappavec_\varsigma } \neq \bar{ \v}_\varsigma \times \bar{ \kappavec }_\varsigma $ unless all dislocation are nicely aligned GNDs (the overbar indicates spatial or ensemble averaging). The reason is that the average of the cross product expression yields the total plastic slip rate which derives from the total dislocation density while the averaged GND-vector $ \kappavec_\varsigma $ usually lacks the total dislocation density information.

To solve this problem there were suggested several workarounds before, as for example: \citet{sedlacek_kw03oa} proposed to work with so called multivalued fields, representing, for example, separable pure GND configurations stemming from multiple sources. This is a useful approach in special situations but it remains of limited scope. \citet{acharya_r06} proposed to split the product expression from the averaging procedure and to substitute the unknown rest with a phenomenological evolution law for the plastic slip rate tensor $ \dot{\mathbfit{L}}^{\mathrm{p}} $ such that
\begin{equation} \label{Eq: evolution acharya}
	\partial_t \bar{\alphabf}_\varsigma = \curl \left( \bar{\v}_\varsigma \times \bar{\alphabf}_\varsigma + \dot{\mathbfit{L}}^{\mathrm{p}} \right).
\end{equation}
While this is a useful approach for deriving effective material laws this does not really add to the question of how to obtain a closed theory from averaging. There have been several attempts to build averaged theories based on the evolution of densities of edge and screw dislocations of two different signs, see for example \citet{arsenlis_etal04, zaiser_h06, reuber_etal14}. These models are kinematically self-contained in the sense that the density information suffices for deriving the plastic slip rate $ \dot{\gamma_\varsigma} $ and the density evolution including dislocation fluxes on the system in closed form. However, they suffer from three drawbacks, that are (\emph{a}) the artificial consideration of only edge and screw character, (\emph{b}) the ignorance to changes in character of the dislocations, and (\emph{c}) the introduction of cumbersome rules for line length changes during the motion of the dislocations. \citet{elazab00}, essentially taking up an idea of \citet{kosevich79}, developed a statistical mechanics theory of dislocation systems rooted in a higher dimensional phase space containing the line direction and the velocity as independent variables. Due to the density description of oriented segments the cancellation problem is successfully solved by such an approach. But the theory as developed by \citet{elazab00} turned out to be a description of unconnected line-segments, a drawback which could be resolved by \citet{hochrainer_zg07} through the introduction of a tensorial measure on the higher dimensional phase space (without the velocity as independent variable). The tensorial measure, called the second order dislocation density tensor (SODT) in the sequel, is a natural generalization of the Kröner-Nye tensor to the higher dimensional configuration space and the evolution equation for it likewise generalizes Eq.\ \eqref{Eq: mura}. However,  working on a higher dimensional configuration space is of course problematic in practical calculations \citep{sandfeld_etal10}. A simplified version of the higher dimensional theory was developed in a series of papers \citep{hochrainer_zg09,sandfeld_etal11} until it was understood as a low order closure to a Fourier expansion of the higher order dislocation density in \citet{hochrainer_etal14}. This simplified CDD contains an evolution equation for the total dislocation density $ \rho_\varsigma $, the classical dislocation density vector $ \kappavec_\varsigma $ and a so-called curvature density $ q_\varsigma $. The simplified theory recovers important features of the higher dimensional CDD and numerical results obtained with it were already successfully compared to discrete dislocation simulations. Nevertheless, it seemed unsatisfactory to develop a theory from a Fourier expansion which seems to suggest a spurious role of the chosen coordinate system. Though it is well known it is not obvious that this multipole expansion corresponds to an expansion of the higher-dimensional distribution function into a series of symmetric traceless tensors of increasing order \citet{applequist89, zheng_z01}. The interpretation as a tensor expansion removes the flavor of being bound to a given coordinate system because tensors are invariantly defined and follow well known transformation rules. Such tensors are known in other branches of continuum mechanics as \emph{alignment tensors}. The notion of alignment tensors was introduced within the theory of liquid crystals \citep{hess75} and finds further applications, e.g.\, in theories for polymers \citep{kroeger1998332} and fiber reinforced composites \citep{advani_t87}. First steps towards the use of alignment tensors in dislocation theory have been presented in \citet{hochrainer13b, hochrainer14}. In the current paper we will present the tensor expansion including the time evolution in full generality.

The outline of the paper is as follows: the tensor expansion provides a hierarchy of tensors which describe the local dislocation state of a crystal in increasing detail. The conservation law for dislocation lines on the higher dimensional space (cf. \citet{hochrainer_zg07}) leads to a hierarchy of evolution equations for the alignment tensors. These evolution equations are themselves expressions of dislocation conservation. They constitute a hierarchy in the sense that the evolution of a tensor of $n$-th order requires information on higher order tenors. But an infinite list of tensors is not a useful set of state variables. Useful theories may only be based on a few low order tensors. To obtain a closed set of evolution equations for low order tensors we need to express the higher order tensors appearing in the evolution equation based on information contained in the low-order tensors. This closing is exemplified for the two lowest order cases. Each time the closing is based on the case where all dislocations are GNDs, which always needs to be recovered from averaged theories as a special case. This work includes the derivation of the evolution equation of the second order alignment tensor presented in \citet{hochrainer14}.

The paper is structured as follows: in Section \ref{Sec: Notation} we introduce some notations and mathematical preliminaries for working on the higher dimensional configuration space. The higher dimensional dislocation density theory is recalled in Section \ref{Sec: hCDD}. In Section \ref{Sec: Multipole expansion} we introduce the multipole expansion in terms of vector spherical harmonics and Fourier series and essentially equivalent expansions into series of irreducible and reducible alignment tensors. In either case we distinguish the general three-dimensional case from the case of only planar dislocations confined to their slip planes. Section \ref{Sec: Evolution alignment} presents the full hierarchy of evolution equations for the multipole expansions and in terms of the alignment tensors. After briefly discussing a similar expansion for the higher dimensional velocity field in Section \ref{Sec: Velocity expansion}, the termination of the hierarchy of evolution equations at low order is topic of Section \ref{Sec: Closure approximations}. The results are discussed and put into the wider context of dislocation density modeling in Section \ref{Sec: Discussion}.

\emph{Remark}: Note that in this Introduction we employed some notations deviating from the use in the rest of the paper. For example we will leave out the slip system index $ \varsigma$ and the overbar for averaged quantities in the sequel. Moreover, the plain symbols $ \rho $ and $ \l $ will be reserved for similar objects on a higher dimensional configuration space.
\section{Notation and the structure of the configuration space} \label{Sec: Notation}
Because we are dealing with a multitude of different mathematical objects we do not strictly apply common rules of using specific classes of symbols for scalars, vectors and tensors. But as usual we denote scalar objects with italic lowercase symbols, and tensorial objects (which includes vectors) with bold faced symbols. However, upper or lowercase bold face symbols do not necessarily indicate a specific type of tensorial object. The frequent use of coordinate notation will usually enlighten the nature of the objects dealt with. It is important to note that although we only deal with the small deformation case, we will distinguish co- and contravariant vectors by means of lower and upper indices, respectively. Inspired by the use in differential geometry (cf. \citet{marsden_h83}) where (contravariant) vectors are identified with partial differential operators we use boldface partial differential symbols to denote the local basis vectors, $\bpartial_i$; however, we use lightface partial differential symbols, $ \partial_i f $, when they operate on functions $f$ as actual partial derivatives. The basis one-forms dual to $\bpartial_i$ are denoted by $\dx{i}$ without using boldface.\footnote{An often employed notation for these bases in continuum mechanics (e.g. \citet{Holzapfel00}) is $\g_i  = \bpartial_i $ and $ \g^i = \dx{i} $.} The Einstein summation convention strictly applies to contracting pairs of one lower and one upper index.

In this paper we deal with objects on the configuration space $ U = M \times S^2 $ made up of spatial points in the crystal manifold $M$ and directions viewed as points on the unit sphere $S^2$. In this case we distinguish the Cartesian spatial parametrization by the co-ordinate functions $x^1$, $x^2$ and $x^3$ and the spherical coordinates with polar angle $\theta$ (north pole in 3-direction) and azimuthal angle $\phi$ taken from the 1-axis. In index notation we distinguish spatial (or horizontal) and directional (or vertical) coordinates (cf. \citet{svendsen01}) by introducing Latin indices $ i,j,k, \ldots \in \lbrace 1,2,3 \rbrace$ and Greek indices $ \mu, \nu, \ldots \in \lbrace \theta,\phi \rbrace$, respectively, such that a tangent vector $ \mathbfit{Y} $ to the configuration space is written as
\begin{equation}
	\mathbfit{Y} = Y^i \bpartial_i + Y^\mu \bpartial_\mu = Y^1 \bparx + Y^2 \bpary + Y^3 \bparz + Y^\theta \bpartial_\theta + Y^\phi \bpartial_\phi.
\end{equation}
The natural metric on the unit tangent bundle accordingly splits into a horizontal and a vertical part as
\begin{equation}
	\g = g_{ij} \dx{i} \otimes \dx{j} + g_{\mu\nu} \dex \mu \otimes \dex \nu,
\end{equation}
where $g_{ij} = \delta_{ij}$ are the Cartesian metric coefficients which equal the Kronecker symbol, while $ g_{\theta \theta} = 1$, $ g_{\theta \phi} = g_{\phi \theta} =  0 $ and $ g_{\phi \phi} = \sin^2 \theta $ are the metric coefficients on the unit sphere. As a consequence of the `block-matrix' structure of the metric, the volume element $ \dVS $ on the configuration space $U$ has a multiplicative structure
\begin{equation}
	\dVS = \dV \wedge \dS,
\end{equation}
where $\dV$ denotes the standard spatial volume element and $ \dS = \sin \theta \, \dex \theta \wedge \dex \phi $ is the solid angle element on the sphere.

In the important special case when dislocations are confined to conservative motion within their glide plane, the configuration space reduces to $ U = M \times S^1 $, with $ S^1 $ being the unit circle representing directions within the glide plane. In that case we usually assume a coordinate system adapted to the glide system, such that the slip direction (normalized Burgers vector $\b$) $\m = \bparx $ points in 1-direction, the glide plane normal $ \n = \bparz $ is the 3-direction and $ \n \times \m = \bpary$ is the 2-direction. The circle is parametrized by the angle $ \phi$ taken from the 1-axis. As there is only one angular direction in this case we do not need free Greek indices, i.e., a tangent vector $ \mathbfit{Y} $ has the form
\begin{equation}
	\mathbfit{Y} = Y^i \bpartial_i + Y^\phi \bpartial_\phi = Y^1 \bparx + Y^2 \bpary + Y^3 \bparz + Y^\phi \bpartial_\phi.
\end{equation}
As in earlier publications we will in this setting usually introduce an extra symbol for the last component function of a vector field, for instance $ \eta = Y^\phi $, although this disguises the vectorial nature of the angular part. In the pure glide scenario the metric and volume element read
\begin{align}
	\g &= g_{ij} \dx{i} \otimes \dx{j} + \dex \phi \otimes \dex \phi, \text{   and} \\
	\dVS &= dV \wedge \dex \phi.
\end{align}

On the configuration spaces we have each time a so called canonical vector field $\l$, which we usually call -- for reasons which will become apparent in the next Section -- the canonical line direction. This vector field is given through the horizontal unit vector pointing in the direction defined through the angular coordinates $ \theta $ and $ \phi $. The line direction in the general case is
\begin{equation}
	\l = \l \left( \theta, \phi \right) = \sin \theta \cos \phi \bparx + \sin \theta \sin \phi \bpary + \cos \theta \bparz.
\end{equation}
That is, the line direction is the position vector of the points on the sphere. In the case of conservative motion we have the canonical line direction as the position vector of the points on the unit circle,
\begin{equation}
	\l = \l \left( \phi \right) = \cos \phi \bparx + \sin \phi \bpary.
\end{equation}
In both cases the canonical line direction is a horizontal vector field which depends on the angular coordinates. In the general case there are two further such vector fields which arise naturally, namely
\begin{align}
	\e_{\theta} &:= \partial_\theta \l = \cos \theta \cos \phi   \bparx + \cos \theta \sin \phi  \bpary - \sin \theta \, \bparz, \\
	\e_{\phi} &:= \partial_\phi \l = - \sin \theta \sin \phi   \bparx + \sin \theta \cos \phi \bpary.
\end{align}
In the case of pure glide we only have the additional base field
\begin{equation}
	\e_{\phi} := \partial_\phi \l = - \sin \phi   \bparx + \cos \phi \bpary.
\end{equation}
In the latter case we will also occasionally use the notation $ \l^\perp = - \e_{\phi} $ as introduced in earlier publications \cite{hochrainer13b, hochrainer_etal14}. The vector fields $ \{ \l, \e_{\phi}, \e_{\theta} \} $ in the general case and $ \{ \l, \e_{\phi}, \bparz \} $ in the pure glide case form an alternative (to the standard basis  $ \{ \bparx, \bpary, \bparz \} $) orthogonal basis of the so called horizontal sub-space of the tangent space to $ M \times S^2 $ and $ M \times S^1 $ respectively. The base vectors $ \bpartial_\phi $ and $ \bpartial_\theta $ on the other hand span the so-called vertical subspace of the tangent space. This denomination will be discussed after introducing a mapping between horizontal and vertical vectors.
\begin{figure}
\begin{center}
\begin{minipage}[b]{0.32\columnwidth}
\includegraphics[width=1\columnwidth]{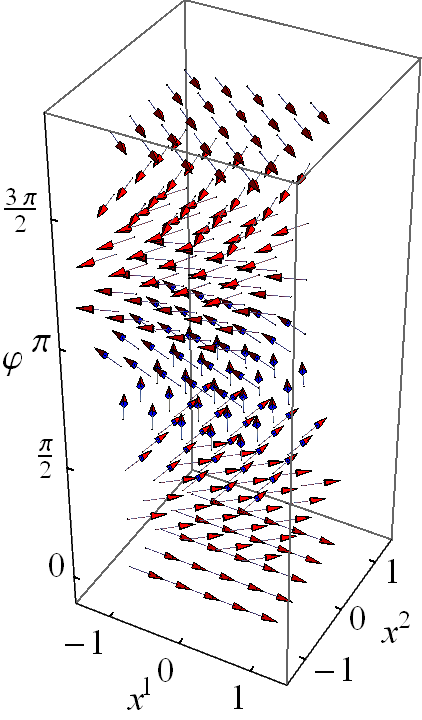} \caption*{\Large{$\l$}}
\end{minipage}
\begin{minipage}[b]{0.32\columnwidth}
\includegraphics[width=1\columnwidth]{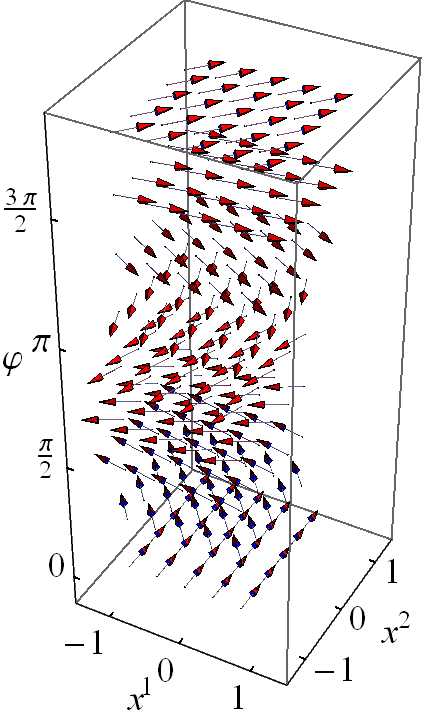} \caption*{\Large{$\e_\phi$}}
\end{minipage}
\begin{minipage}[b]{0.32\columnwidth}
\includegraphics[width=1\columnwidth]{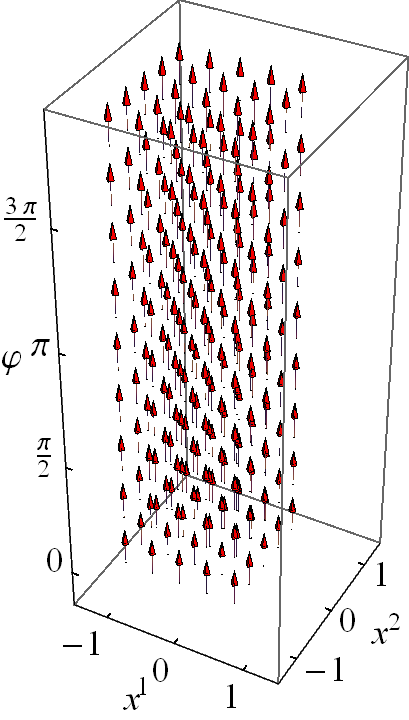} \caption*{\Large{$\bpartial_\phi$}}
\end{minipage}
\caption{Visualization of the basis fields $ \l $ , $\e_\phi$ and $\bpartial_\phi$ in a three dimensional slice of the four dimensional configuration space in the case of glide only.}%
\label{fig_1}
\end{center}
\end{figure}

A mapping from vertical vectors $\etabf = \eta^\mu \bpartial_\mu $ to horizontal vectors can be defined by
\begin{equation}
	\etabf = \eta^\mu \bpartial_\mu \mapsto  \eta^\mu \partial_\mu \l = \eta^\mu \e_{\mu} =: \etabf_\mathrm{h} .
\end{equation}
If restricted to horizontal vectors which are orthogonal to $\l$ (and lie in the glide plane in the case of conservative motion) there is a unique inverse of the map, such that horizontal vectors of the form $ \u = u^\mu \e_{\mu} $ are mapped to vertical vectors through
\begin{equation}
	\u = u^\mu \e_{\mu} \mapsto  u^\mu \bpartial_\mu =: \u_\mathrm{v}.
\end{equation}
In Figure \ref{fig_1} the vector fields $ \l $ , $\e_\phi$ and $\bpartial_\phi$ are visualized in a three dimensional slice (disregarding the normal direction $x^3$ to the glide plane) of the four dimensional configuration space in case of glide only. The $x^1$ and $x^2$ directions span the glide plane, while the third dimension shown is the angular direction parametrized by $ \phi $. One easily sees that the three vector fields are pointwise orthogonal to each other. This visualization also motivates the denomination of horizontal vector fields (tangent to the base-manifold) and vertical ones (connected to the extra dimension). Moreover this helps realizing that the mapping from horizontal to vertical vector fields, i.e., mapping $ \e_\phi $ into $\bpartial_\phi$ in the pure glide case, actually involves very distinct objects. This needs to be kept in mind also when visualizing the base fields in the five-dimensional case as attempted in Figure \ref{fig_2}. In that figure the spatial basis $\{ \l, \e_{\phi}, \e_{\theta} \}$ is depicted within the sphere. The canonical line direction $ \l$ points towards the current point on the sphere along the red ray (color version is available online). The base vectors $ \e_{\phi} $ and $ \e_{\theta} $ apparently coincide with the vertical base vectors $ \bpartial_\phi $ in direction of the azimuthal and $ \bpartial_\theta $ in direction of the polar angle, respectively, which are depicted tangent to the sphere; note, however, that the latter base vectors actually belong to two further dimensions which can not be visualized.
\begin{figure}
\begin{center}
\includegraphics[width=0.5\columnwidth]{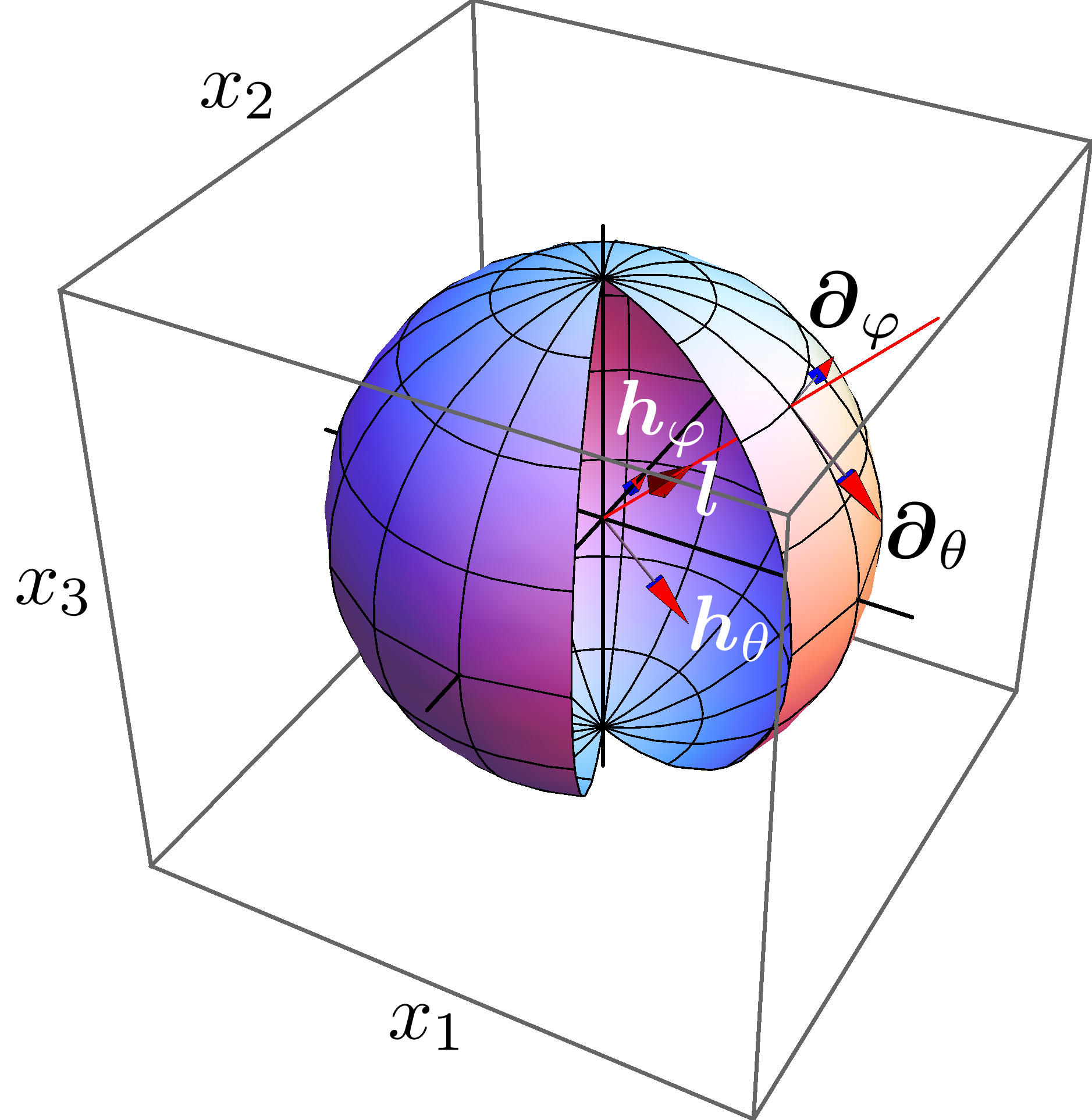}
\caption{Visualization of the horizontal basis $\{ \l, \e_{\phi}, \e_{\theta} \}$ (within the sphere) and the vertical basis $\{ \bpartial_\phi, \bpartial_\theta \}$ (tangent to the sphere) at a point in the sphere bundle. }%
\label{fig_2}
\end{center}
\end{figure}
\subsection{Differential forms on the configuration space} \label{Sec: Diff forms}
As discussed in the last section we deal with objects on five and four dimensional spaces. As a consequence, field theoretic descriptions cannot make use of classical vector calculus but are most conveniently formulated in terms of differential forms. Although the use of differential forms is not wide spread we refer to standard text books (e.g., \cite{do_carmo94, marsden_h83}) for an introduction to the formalism. The reader should be familiar with (\emph{i}) the wedge product $ \wedge $, which generalizes the cross and the scalar triple product, (\emph{ii}) the interior multiplication $ \i_\v $ with a vector field $ \v $, which generalizes the scalar product and also the cross product, (\emph{iii}) the exterior derivative $ \dex $, which generalizes the gradient, the curl and the divergence operators, (\emph{iv}) the Lie-derivative along a vector field $ \Lie_\v = \dex \i_\v + \i_\v \dex $, which provides operators appearing in conservation laws of differential forms and (\emph{v}) the generalized Stokes theorem. The only mathematical tool discussed here is the projection of differential forms on the configuration spaces onto spatial differential forms by integration. This is a generalization of what is called Fubini's theorem by \citet{svendsen01}.

The configuration spaces $ M \times S^2 $ and $ M \times S^1 $ are product manifolds. We may therefore project differential forms to purely spatial or basal differential forms by integration over the directional space. We regard a $p$-form  $ \mathbfit\omega $ on the configuration space, where $p \geq 2$ for general configurations and $ p \geq 1 $ in the case of only planar dislocations. Let $\u_{(j)} = u^i_{(j)} \bpartial_i$ denote vector fields on the base manifold, that is, such that the coordinate functions $ u^i_{(j)} $ are independent of the angular coordinates $ \theta $ and $\phi $. In the general configuration space we take $p-2$ such vectors to obtain from $ \mathbfit\omega $ by repeated interior multiplication the differential $2$-form $ \i_{\u_{(j_{p-2})}} \cdots \i_{\u_{(j_1)}} \mathbfit\omega $ on the configuration space. This $2$-form can be integrated over two-dimensional subspaces as the sphere at every point. We then define a $(p-2)$-form on the base manifold through
\begin{equation} \label{Eq: base projection}
	\mathbfit\omega^\mathrm{b} \left( \u_{(j_1)}, \cdots, \u_{j_{(p-2)}} \right) = \int_{S^2} \i_{\u_{(j_{p-2})}} \cdots \i_{\u_{(j_1)}} \mathbfit\omega
\end{equation}
In local coordinates, $\mathbfit\omega$ is of the form
\begin{equation}
	\mathbfit\omega = \omega_{i_1 \cdots i_p} \dx{i_1} \wedge \ldots \wedge \dx{i_p} + \omega_{i_1 \cdots i_{p-1} \mu} \dx{i_1} \wedge \ldots \wedge \dx{i_{p-1}} \wedge \dex \mu + \omega_{i_1 \cdots i_{p-2} \mu \nu} \dx{i_1} \wedge \ldots \wedge \dx{i_{p-2}} \wedge \dex \mu \wedge \dex \nu.
\end{equation}
Only the last term of this contributes to the integral in (\ref{Eq: base projection}) such that we find
\begin{equation}
	\omega^\mathrm{b}_{i_1 \cdots i_{p-2}} = \mathbfit\omega^\mathrm{b} \left( \bpartial_{i_1}, \cdots, \bpartial_{i_{p-2}} \right) = \int_{S^2} \omega_{i_1 \cdots i_{p-2} \mu \nu} \, \dex \mu \wedge \dex \nu.
\end{equation}
This justifies the short hand notation
\begin{equation}
	\mathbfit\omega^\mathrm{b} =  \int_{S^2} \mathbfit\omega = \int_{S^2} \left( \omega_{i_1 \cdots i_{p-2} \mu \nu} \dex \mu \wedge \dex \nu \right) \dx{i_1} \wedge \ldots \wedge \dx{i_{p-2}}.
\end{equation}
Because the sphere is a closed manifold we find that this basal projection $ \left( \cdot \right)^\mathrm{b} $ commutes with the exterior derivative $ \dex $, i.e., $ \left( \dex \mathbfit\omega\right) ^\mathrm{b} = \dex \mathbfit\omega^\mathrm{b} $. To see this we take $ U \subseteq M $ to be a submanifold of dimension $ p+1 $, meaning that the integral of $ \left( \dex \mathbfit\omega\right) ^\mathrm{b} $ over $ U $ is well defined. For this we obtain from Stokes theorem
\begin{equation}
	\int_{U} \left( \dex \mathbfit\omega\right) ^\mathrm{b} = \int_{U} \int_{S^2} \dex \mathbfit\omega = \int_{U \times S^2}  \dex \mathbfit\omega = \int_{\partial \left(U \times S^2\right)}  \mathbfit\omega = \int_{\partial U \times S^2} \mathbfit \omega = \int_{\partial U} \int_{S^2} \mathbfit \omega = \int_{\partial U} \mathbfit \omega^\mathrm{b} = \int_{U} \dex \mathbfit \omega^\mathrm{b}.
\end{equation}
Because this integral equality holds for any submanifold $ U $ of dimension $ p+1 $ this requires the integrands to be the same, i.e. that $ \left( \dex \mathbfit\omega\right) ^\mathrm{b} = \dex \mathbfit\omega^\mathrm{b} $.

With the proper dimensional changes the basal projection is analogously defined for the case of glide only, where the configuration space is $ M \times S^1 $. In this case we have
\begin{equation}
	\mathbfit\omega^\mathrm{b} =  \int_{S^1} \mathbfit\omega = \int_{S^1} \left( \omega_{i_1 \cdots i_{p-1} \phi} \dex \phi \right) \dx{i_1} \wedge \ldots \wedge \dx{i_{p-1}}.
\end{equation}
The basal projection in the case of glide only likewise commutes with the exterior derivative because also $ S^1 $ is a closed manifold.
\subsection{Vector valued differential forms} \label{Sec: Vector valued}
Besides the differential forms discussed above we also deal with tensor valued differential forms on the configuration space. Spatial projections of such tensor valued differential forms on the configuration space will in the sequel define the tensor expansion of the higher dimensional dislocation density variables. Even though these differential forms are tensor valued one usually denotes them as vector valued because the decisive point is that they take values in a vector bundle. We refer to the appendix in \citet{hochrainer13b} for a brief introduction to the calculus on vector valued differential forms. In the current case the differential forms take values in tensor bundles built upon the tangent space to the base manifold $ TM $. In the case of small deformations this is a trivial vector bundle, that is, it may be identified with a product manifold\footnote{We note that the flat crystal connection also renders the tangent bundle trivial in the case of large deformations, see e.g. \citet{hochrainer13b}}: $ TM \cong M \times \mathbb{R}^3 $. Because of this product structure the vector valued differential forms assume values in a vector space and we may perform differential and integrals calculus for the tensor coefficients in a fixed basis. This is what one routinely does when integrating vector valued fields, as. e.g. force fields in mechanics. The exterior derivative $ \dex $, the basal projection $ \left( \cdot \right)^\mathrm{b} $ and Stokes theorem may thence be defined coefficient-wise. Most notably, the basal projection and the exterior derivative also commute for the tensor valued differential forms. For readers who are interested in further details of the vector valued forms we provide a quick overview of the discussed concepts below. The biggest part of these details may be skipped without impairing the readability of the rest of the paper, with exception of Eq.\ \eqref{Eq: exterior VVDF} below.

Let $ \mathbfit{\Omega} \in \Lambda^p (M \times S^2 ) \otimes T^r M $ be a differential $p$-form on the configuration space assuming purely contravariant $r$-th order tensors on the base manifold as values, i.e.
\begin{multline}
	\mathbfit{\Omega} = \tensor{\Omega}{_{i_1 \cdots i_{p} \mu \nu}^{j_1 \cdots j_r}} \dex x^{i_1} \wedge \ldots \wedge \dex x^{i_p} \otimes \bpartial_{j_1} \otimes \cdots \otimes \bpartial_{j_r} + \\
	\tensor{\Omega}{_{i_1 \cdots i_{p-1} \mu }^{j_1 \cdots j_r}} \dex x^{i_1} \wedge \ldots \wedge \dex x^{i_{p-1}} \wedge \dex \mu \otimes \bpartial_{j_1} \otimes \cdots \otimes \bpartial_{j_r} + \\
	\tensor{\Omega}{_{i_1 \cdots i_{p-2} \mu \nu}^{j_1 \cdots j_r}} \dex x^{i_1} \wedge \ldots \wedge \dex x^{i_{p-2}} \wedge \dex \mu \wedge \dex \nu \otimes \bpartial_{j_1} \otimes \cdots \otimes \bpartial_{j_r}.
\end{multline}
In analogy to the definition for differential forms we define the basal projection
\begin{equation}
	\mathbfit{\Omega}^\mathrm{b}  = \int_{S^2} \left( \tensor{\Omega}{_{i_1 \cdots i_{p-2} \mu \nu}^{j_1 \cdots j_r}} \, \dex \mu \wedge \dex \nu \right) \dex x^{i_1} \wedge \ldots \wedge \dex x^{i_{p-2}} \otimes \bpartial_{j_1} \otimes \cdots \otimes \bpartial_{j_r}.
\end{equation}
Also this projection commutes with the exterior derivative on vector valued differential forms. In order to show this we first discuss how the exterior derivative is defined on vector valued differential forms and then how this yields a generalization of Stokes theorem for the vector valued case. The exterior derivative on vector valued differential forms is defined using the Levi-Civita-connection on the vector bundle $ \nabla $ by requiring the following product rule for decomposable vector valued differential forms of the form $ \mathbfit{\Omega} = \mathbfit{\omega} \otimes T^r $:
\begin{equation} \label{Eq: exterior VVDF}
	\dex \left( \mathbfit{\omega} \otimes T^r  \right) = \dex \mathbfit{\omega} \otimes T^r + (-1)^p \mathbfit{\omega} \wedge \nabla T^r,
\end{equation}
where $ \nabla T^r $ is considered an $r$-tensor valued $1$-form. In order to formulate a generalized Stokes theorem we need to define integrals of the tensor valued differential forms. Because of the trivial structure of the tensor bundles we may view the tensor-valued differential form as a collection of $ 3^r $ differential forms
\begin{multline}
	\mathbfit{\omega}^{j_1 \cdots j_r } = \tensor{\Omega}{_{i_1 \cdots i_{p} \mu \nu}^{j_1 \cdots j_r}} \dex x^{i_1} \wedge \ldots \wedge \dex x^{i_p} + 
	\tensor{\Omega}{_{i_1 \cdots i_{p-1} \mu }^{j_1 \cdots j_r}} \dex x^{i_1} \wedge \ldots \wedge \dex x^{i_{p-1}} \wedge \dex \mu + \\
	\tensor{\Omega}{_{i_1 \cdots i_{p-2} \mu \nu}^{j_1 \cdots j_r}} \dex x^{i_1} \wedge \ldots \wedge \dex x^{i_{p-2}} \wedge \dex \mu \wedge \dex \nu.
\end{multline}
We use these differential forms to define tensor components from integrals over a $p$-dimensional submanifold $\mathscr{U} \subseteq M \times S^2 $ in a fixed basis through 
\begin{equation}
T_\mathscr{U} ^{\mathbfit{\Omega} \; j_1 \cdots j_r} =  \int_\mathscr{U} \mathbfit{\omega}^{j_1 \cdots j_r },
\end{equation}
and define the tensor valued integral of the vector valued differential form as
\begin{equation}
  \int_\mathscr{U} \mathbfit\Omega := T_\mathscr{U} ^{\mathbfit{\Omega} \; j_1 \cdots j_r} \bpartial_{j_1} \otimes \cdots \otimes \bpartial_{j_r}.
\end{equation}
For each $p$-form $\mathbfit{\omega}^{j_1 \cdots j_r }$ Stokes theorem applies, such that for a $ p+1 $ dimensional submanifold $ \mathscr{U} $ we find
\begin{equation}
T_\mathscr{U} ^{\dex \mathbfit{\Omega} \; j_1 \cdots j_{r}} =  \int_\mathscr{U} \dex \mathbfit{\omega}^{j_1 \cdots j_r } =  \int_{\partial \mathscr{U}} \mathbfit{\omega}^{j_1 \cdots j_r } = T_{\partial \mathscr{U}} ^{\mathbfit{\Omega} \; j_1 \cdots j_{r}},
\end{equation}
such that we obtain Stokes theorem for vector valued differential forms
\begin{equation}
  \int_\mathscr{U} \dex \mathbfit\Omega = \int_\mathscr{\partial U} \mathbfit\Omega.
\end{equation}
With Stokes theorem at hand we thus obtain with the same calculation as for the real valued differential forms ($U \subseteq M$ being again a $p+1$-dimensional submanifold)
\begin{equation}
	\int_{U} \left( \dex \mathbfit\Omega\right) ^\mathrm{b} = \int_{U} \dex \mathbfit\Omega^\mathrm{b},
\end{equation}
and similarly conclude that the basal projection commutes with the exterior derivative also for vector valued differential forms, i.e. that $ \left( \dex \mathbfit\Omega\right) ^\mathrm{b} = \dex \mathbfit\Omega^\mathrm{b} $.
\section{Higher dimensional continuum dislocation dynamics } \label{Sec: hCDD}
At the heart of the mesoscopic continuum theory of dislocations developed by \cite{hochrainer_zg07} lies the so called dislocation density tensor of second order (SODT) $ \alphaII $. This tensor is a natural generalization of the classical dislocation density tensor to a higher dimensional configuration space, defined as direct product of the real space and the space of possible line directions. The tensor is closely related to the phase space densities of dislocations as introduced by \cite{elazab00}, but as tensor it explicitly accounts for the curved and connected nature of dislocation lines. The SODT is a vector valued 4-form on the five-dimensional space $M\times S^2$ and a vector valued 3-form on the four-dimensional space $M\times S^1 $ in the case of glide only. Dislocations have a constant Burgers vector which stays unchanged while moving. This means that if we obtain a density measure solely for dislocations of a given Burgers vector, the evolution of the density measure does not involve transfer to or from densities of another Burgers vector. We may thus split the density measure in its contributions from given Burgers vectors in the general case and for given glide systems in the case of glide only. The general \emph{kinematic} theory then derives by summing over all slip systems.\footnote{The \emph{dynamic} theory will of course have to consider interactions between dislocations of different Burgers vector or on distinct glide planes. We will focus on the kinematic theory in most of the paper and only touch upon the huge challenge of the dynamic case in Section \ref{Sec: Discussion}.} This means that from now on we consider dislocations of a given Burgers vector, $ \b $, or a given glide system $\n, \b$ only. The SODT is then multiplicatively decomposable in its differential form part $ \kappaII $ and the Burgers vector as
\begin{equation}
 	\alphaII = \kappaII  \otimes \b.
\end{equation}
Due to this structure the evolution of the SODT is solely determined by the evolution of the differential form $\kappaII$ which contains the geometric information of the dislocation distribution. We note that the Burgers vector is needed to determine the Peach-Koehler force and thus the dislocation velocity, but it plays no role for the kinematics of the dislocation density description once the dislocation velocity $ \v $ is known.

Before we turn to the evolution equation, however, we provide a less abstract form for $\kappaII$. In either setting this is a differential $n-1$ form on an $n$ dimensional space. Consequently, there is each time a unique associated vector field $ \R $ such that
\begin{equation}
 	\kappaII = \i_\R \dVS.
\end{equation}
We call $ \R $ the dislocation density vector of second order. It represents a density of lifted dislocation line directions on the configuration space \citep{hochrainer_zg07}. As each tangent vector to the configuration space $ \R $ can be split into its horizontal and vertical part. Because the horizontal part of the tangent to the lifted dislocation lines is the canonical line direction $\l$ the horizontal part is given by $\rho \l$, with a density function $\rho$. We denote the vertical part with $ \q $. The dislocation density vector of second order consequently takes the form
\begin{equation} \label{Eq: R}
   \R = \rho \sin \theta \cos \phi \bparx + \rho \sin \theta \sin \phi \bpary +\rho \cos \theta \bparz + q^\theta \bpartial_{\theta} + q^\phi \bpartial_{\phi},
\end{equation}
in the general case and
\begin{equation} \label{Eq: R cons}
   \R = \rho \cos \phi \bparx + \rho \sin \phi \bpary +  q \bpartial_{\phi},
\end{equation}
in the case of only conservative motion. Accordingly, the differential form part of the SODT is given by
\begin{multline} \label{Eq: kappaII}
   \kappaII = \left( \rho \sin \theta \cos \phi \dx{2} \wedge \dx{3} - \rho \sin \theta \sin \phi \dx{1} \wedge \dx{3} + \rho \cos \theta \dx{1} \wedge \dx{2} \right) \wedge \dS + \\ \sin \theta \dV \wedge \left( - q^\theta \dex \phi + q^\phi \dex \theta \right),
\end{multline}
and
\begin{equation} \label{Eq: kappaII cons}
   \kappaII = \left( \rho \cos \phi \dx{2} \wedge \dx{3} - \rho \sin \phi \dx{1} \wedge \dx{3} \right) \wedge \dex \phi  - q \dV,
\end{equation}
respectively. From this differential form on the configuration space we can obtain the spatial dislocation density 2-form $ \kappabf $ as the spatial projection, i.e.,
\begin{equation} \label{Eq: kappa}
   \kappabf = \int_{S^2} \left( \rho \sin \theta \cos \phi \, \dS \right) \dx{2} \wedge \dx{3} - \int_{S^2} \left( \rho \sin \theta \sin \phi \, \dS \right) \dx{1} \wedge \dx{3} +  \int_{S^2} \left( \rho \cos \theta \, \dS \right) \dx{1} \wedge \dx{2},
\end{equation}
or, in case of glide only,
\begin{equation} \label{Eq: kappa cons}
   \kappabf = \int_{S^1} \kappaII = \int_{S^1} \left( \rho \cos \phi \, \dex \phi \right) \dx{2} \wedge \dx{3} - \int_{S^1} \left( \rho \sin \phi \, \dex \phi  \right) \dx{1} \wedge \dx{3}.
\end{equation}
Also for the spatial differential 3-form $\kappabf $ we may define a vector $ \kappavec $ such that $ \kappabf = \i_{\kappavec} \dV $ with
\begin{equation} \label{Eq: kappavec}
   \kappavec = \int_{S^2} \left( \rho \sin \theta \cos \phi \, \dS \right) \bparx + \int_{S^2} \left( \rho \sin \theta \sin \phi \, \dS \right) \bpary +  \int_{S^2} \left( \rho \cos \theta \, \dS \right) \bparz,
\end{equation}
or 
\begin{equation} \label{Eq: kappavec cons}
   \kappavec = \int_{S^1} \left( \rho \cos \phi \, \dex \phi \right) \bparx + \int_{S^1} \left( \rho \sin \phi \, \dex \phi  \right) \bpary.
\end{equation}
Before we turn to the time evolution we note that also the lifted dislocation lines do not end inside the configuration space and therefore $\kappaII$ (and consequently $\alphaII$) turns out to be a closed differential form, i.e.,
\begin{equation} \label{Eq: d kappaII}
   \dex \kappaII = 0.
\end{equation}
This likewise means that $ \R $ is solenoidal. If we denote with $\nabla^S \cdot $ the divergence on the sphere, we find from Eq.\ \eqref{Eq: R} the solenoidality of $ \R $ to be equivalent to
\begin{equation} \label{Eq: d kappaII coord}
   \dex \kappaII = 0 \Leftrightarrow \nabla_\l \rho + \nabla^S \cdot \q = 0.
\end{equation}
For glide only systems this turns into the condition that $\nabla_\l \rho + \parphi q = 0$. Because the projection procedure commutes with the exterior derivative, we immediately see that also the projected differential form $ \kappabf $ is closed,
\begin{equation} \label{Eq: kappa cons glide only}
   \dex \kappabf = 0.
\end{equation}
This likewise implies that the vector field $ \kappavec $ is solenoidal as required from classical dislocation density theory.
\subsection{Time evolution} \label{Sec: Time evolution}
The motion of the lifted dislocation lines in the configuration space is determined from the motion of the dislocations in the crystal manifold $ M$. The velocity vector on the configuration space $ \V = \v + \thetabf $ is composed of a horizontal vector $ \v $ which is orthogonal to $ \l $ and a vertical part $\thetabf$ which characterizes the rotational direction of dislocation segments while moving. In the current small deformation formulation the horizontal image of the rotational velocity $\thetabf_\mathrm{h} $ is given through the component of the directional derivative of the velocity vector $ \v $ along the lifted line direction $ \L = 1/ \rho \R $, which (the component) is orthogonal to the canonical line direction $\l$, that is
\begin{equation} \label{Eq: theta}
   \thetabf_\mathrm{h} = \left( \nabla_\L \v \right)^\perp= \nabla_\L \v - \g \left(  \nabla_\L \v, \l \right) \l. 
\end{equation}
Consequently $\thetabf = \vartheta^\mu \bpartial_\mu $, where the $\vartheta^\mu = \vartheta^\mu_\mathrm{h}$ are obtained from the representation $ \thetabf_\mathrm{h} = \vartheta^\mu_\mathrm{h} \e_\mu $. We summarize that the lifted velocity vector field is of the form
\begin{equation} \label{Eq: V}
   \V  = \v + \thetabf = v^i \bpartial_i + \vartheta^\mu \bpartial_\mu = v^\mu \e_\mu + \vartheta^\mu \bpartial_\mu,
\end{equation}
where the right most representation of $ \v $ points to the fact that the spatial dislocation velocity is orthogonal to the line direction $\l$.

Again, we introduce a special notation for the pure glide case, where also the velocity vector has only one component in $ \e_\phi $ direction, which we call the scalar velocity. However, in order to be consistent with earlier publications, we introduce a \emph{sign convention} for this scalar velocity such the $ \v = -v \e_\phi = v \left(-\e_\phi \right) $. Note that appended to a point on the unit circle $ \e_\phi $ points to the center of the loop. The sign convention therefore ensures that expanding positively oriented loops will have a positive velocity. In the case of glide only we usually leave out the vector index of the rotational velocity and write $ \thetabf = \vartheta \partial_\phi $. With the sign convention from above we find $ \vartheta = - \nabla_\L v $, such that
\begin{equation} \label{Eq: V cons}
   \V  =  - v \e_\phi + \vartheta \bpartial_\phi = v \l^\perp - \nabla_\L v \, \bpartial_\phi,
\end{equation}
where we remind that $\l^\perp = -\e_\phi $.

Given the lifted dislocation velocity field $ \V $ the time evolution of the SODT is given by a Lie derivative $\Lie$ in the direction of the velocity, i.e.\ 
\begin{equation} \label{Eq: dt kappaII}
   \partial_t \kappaII  =  - \Lie_\V \kappaII =  - \dex \i_\V \kappaII.
\end{equation}
The right most expression follows from Cartan's magic formula $ \Lie_\V = \dex \i_\V + \i_\V \dex $ and the fact that the SODT is closed, $ \dex \kappaII = 0 $. The simplicity of evolution equation Eq.\ \eqref{Eq: dt kappaII} is one important reason to work with differential forms in the higher dimensional theory. From the higher dimensional evolution equation we obtain the evolution equation for the spatial dislocation density two-form as
\begin{equation} \label{Eq: dt kappa}
   \partial_t \kappabf  =  - \int_{S^{2(1)}} \dex \i_\V \kappaII = - \dex \int_{S^{2(1)}} \i_\V \kappaII.
\end{equation}
The term $ \J = \i_V \kappaII $ is the higher dimensional dislocation flux form and from the last equation we immediately see that its spatial projection is the classical dislocation flux 1-form, which upon multiplication with the Burgers vector yields the plastic distortion rate tensor $\partial_t \mathbfit\beta^\mathrm{p}$ ,
\begin{equation} \label{Eq: dt beta}
   \partial_t \mathbfit\beta^\mathrm{p}  = - \int_{S^{2(1)}} \i_\V \kappaII \otimes \b = - \int_{S^{2(1)}} \left( \rho \varepsilon_{ikj} v^j l^k \dS \right) \dx{i} \otimes \b.
\end{equation}
Note that in the case of glide only Eq.\ \eqref{Eq: dt beta} takes the well known form 
\begin{equation} \label{Eq: dt beta glide only}
  \partial_t \mathbfit\beta^\mathrm{p}  = \int_{S^{1}} \rho v \dS \n \otimes \b.
\end{equation}
With this stipulation and upon translating the equations into classical vector calculus we obviously recover the rate form of what Kr\"oner called the fundamental equation of dislocation theory, i.e.\ $ \partial_t \alphabf  = \curl \left( \partial_t \mathbfit\beta^\mathrm{p} \right) $.

\section{Multipole expansion of the SODT} \label{Sec: Multipole expansion}

In the last Section we introduced the higher dimensional theory as a natural extension of the classical dislocation density theory; the latter being recovered from the SODT and its evolution by spatial projections. The central motivation for the higher dimensional theory is the fact that from the classical dislocation density tensor $\alphabf $ one may only in very limited special cases recover the plastic distortion rate tensor $\partial_t \mathbfit\beta^\mathrm{p}$, as required for a closed theory. This is only possible if dislocations form smooth line bundles, which is a strong limitation in view of typical dislocation configurations. The higher dimensional theory allows a meaningful closure with much weaker restrictions on the represented dislocation distribution, namely if the \emph{lifted} dislocations form smooth line bundles on the configuration space, as would be the case if neighboring segments of the same orientation share the same curvature. This restriction may even be relieved if dislocations move by glide only and the magnitude of the dislocation velocity does not depend on the dislocation character \citep{hochrainer07}. Hence, the higher dimensional theory is an important generalization which allows predicting the evolution of non-trivial dislocation distributions \citep{sandfeld_etal10,sandfeld_etal11}. But even though the higher dimensional theory acts as an ensemble simulation possibly replacing a whole series of discrete dislocation simulations, the numerical effort is still far beyond what could be accepted for engineering applications of small scale structures.

Although it seems obvious in the back view, we only recently found that a multipole expansion of the SODT would pave the way for a systematic reduction of the higher dimensional theory into manageable objects and evolution equations \citep{hochrainer_zg09,sandfeld_etal10,hochrainer13b, hochrainer_etal14}. A multipole expansion using spherical harmonics in general and trigonometric functions (Fourier expansion) in the conservative case at first sight looks rather like a numerical approach for discretizing the equations, as done in (Fast) Fourier methods for solving differential equations. Though it is well known it is not obvious that the multipole expansions each time correspond to an expansion of a higher-dimensional distribution function into a series of symmetric traceless tensors of increasing order \citet{applequist89}. The interpretation as a tensor expansion is vital for a geometrical understanding of the expansion. Furthermore, the tensor interpretation removes the flavor of being bound to a given coordinate system because tensors are invariantly defined and follow well known transformation rules. Such tensor expansions are well known in other branches of materials science dealing with distributions of line like objects as, e.g., in the theory of liquid crystals, polymers or fiber reinforced composites. Within these fields the appearing tensors are called \emph{alignment tensors}, a term which we adopt for the dislocation case. For dislocations, however, there are two important extensions to the established tensor expansions in the above mentioned fields. For most liquid crystals, polymers and short fibers the orientation of the line like objects is irrelevant. This \emph{head-tail-symmetry} leads to density functions on the orientation space which are symmetric with respect to inversion, i.e., $ \rho \left(\l\right) = \rho \left(-\l\right) $. As a consequence, only alignment tensors of even order are non-trivial in these cases. This is not true for dislocations, such that also tensors of odd order will have to be considered. Moreover, the dislocation state on the configuration space is not fully characterized by the density function $ \rho $ alone, but by a 4-form $ \kappaII $ or equivalently the tangent vector field $\R$. As a consequence the SODT must be expanded using vector spherical harmonics rather than scalar spherical harmonics. This likewise leads to an expansion into each time three tensors of a given order in general and a series of each time two tensors of a given order in the case of conservative dislocation motion. One tensor series will be the well known expansion of the density $ \rho $ into alignment tensors while the other series of tensors expand the curvature vector $ \q $. The tensor expansion for vectors is seemingly not so common, but rather straight forward from the scalar tensor expansion and the definition of the vector spherical harmonics.

\subsection{Expansion with vector spherical harmonics and vector Fourier expansion} \label{Sec: Expansion}

Let $ Y_\km \left( \theta, \phi \right) $ denote the scalar \emph{real} (or tesseral) spherical harmonics. Then the orthonormalized vector spherical harmonics (on the unit sphere) in contravariant form are defined as \citep{regge_w57,barrera_eg85}
\begin{align}
	\Y_\km &= Y_\km \l = Y_\km l^i \bpartial_i, \\
	\Psibf_\km &= \frac{1}{\sqrt{k\left(k+1\right)} } \grad Y_\km = \frac{1}{\sqrt{k\left(k+1\right)} } g^{\mu \nu} \partial_\nu Y_\km \bpartial_{\mu} \text{ and} \\
	\Phibf_\km &= \frac{1}{\sqrt{k\left(k+1\right)} } \grad^{*} Y_\km = \frac{1}{\sqrt{k\left(k+1\right)} } \varepsilon^{\mu \nu} \partial_\nu Y_\km \bpartial_{\mu}.
	\end{align}
The coefficients of the contravariant metric tensor on the sphere are given through $g^{\theta \theta} = 1$, $ g^{\theta \phi} = g^{\phi \theta} = 0$ and $g^{\phi \phi}=1/\sin^2 \theta$. Furthermore, the contravariant components of the Levi-Civita tensor are $\varepsilon^{\theta \theta} = \varepsilon^{\phi \phi} = 0$ and $ \varepsilon^{\theta \phi} = -\varepsilon^{\phi \theta} = 1/\sin \theta $.
The vector field may then be expressed as
\begin{equation} \label{Eq: expansion R}
	\R =\sum_{k=0}^\infty\sum_{m=-k}^k\left(\rho_\km \mathbfit{Y}_\km+q_\km \Psibf_\km+q^*_\km \Phibf_\km\right),
\end{equation}
where the coefficients are obtained from
\begin{definition}[Vector spherical harmonics coefficients]\label{VSH coefficients}
\begin{align}
	\rho_\km &:= \int_{S^2} \R \cdot \Y_\km \dS = \int_{S^2} \rho Y_\km \dS\\
	q_\km &:= \int_{S^2} \R \cdot \Psibf_\km \dS = \frac{1}{\sqrt{k\left(k+1\right)} } \int_{S^2} q^\mu \partial_\mu Y_\km \dS\\
	q^*_\km &:= \int_{S^2} \R \cdot \Phibf_\km \dS = \frac{1}{\sqrt{k\left(k+1\right)} } \int_{S^2} q^\mu \tensor{\varepsilon}{_\mu^\nu} \partial_\nu Y_\km \dS.
\end{align}
\end{definition}
We note that $ \tensor{\varepsilon}{_\theta^\theta} = \tensor{\varepsilon}{_\phi^\phi} = 0 $, $ \tensor{\varepsilon}{_\theta^\phi} = 1/\sin \theta$ and $ \tensor{\varepsilon}{_\phi^\theta} = - \sin \theta$. From Equation (\ref{Eq: expansion R}) we immediately derive the expansion of the differential form $\kappaII$ as
\begin{align} \label{Eq: expansion kappa}
	\kappaII &= \sum_{k=0}^\infty\sum_{m=-k}^k\left(\rho_\km \i_{\mathbfit{Y}_\km} \dVS +q_\km \i_{\Psibf_\km} \dVS +q^*_\km \i_{\Phibf_\km} \dVS \right) \\
	&=: \sum_{k=0}^\infty\sum_{m=-k}^k\left(\mathbfit{\kappa}^\Y_\km + \mathbfit{\kappa}^\Psibf_\km + \mathbfit{\kappa} ^\Phibf_\km \right)
\end{align}
For the case of pure glide we note that instead of spherical harmonics we deal with the trigonometric functions as used in Fourier expansion. In order to keep a similar appearance we define the functions
\begin{equation} 
	Y_{\left(m\right)} = \left\{\begin{array}{rl} \cos \left( m \phi \right), & \mbox{if } m > 0 \\ 1, & \mbox{if } m = 0 \\ \sin\left( | m | \phi \right), & \mbox{if } m < 0. \end{array} \right. 
\end{equation} 
Then the vector basis reads\footnote{It seems worth noting that there appears a non-trivial $ \Psibf_{(0)} = \bpartial_{\phi} $ in the two-dimensional case which has no counterpart in the three-dimensional vector spherical harmonics, where $ \Psibf_{(00)} = 0 $ and $ \Phibf_{(00)} = 0 $. The reason is that a `constant' vector field $ \bpartial_{\phi} $ only exists on the unit circle while for topological reasons there are no such fields on the sphere.}
\begin{align}
	\Y_\mm &= Y_\mm \l = Y_\mm l^i \bpartial_i, \\
	\Psibf_\mm &=  \left\{\begin{array}{rl} \bpartial_{\phi}, & \mbox{if } m = 0 \\ \frac{1}{|m|} \partial_\phi Y_\mm \bpartial_{\phi}, & \mbox{if } m \neq 0. \end{array} \right. 
\end{align}
 The expansion of $\R$ then takes the form
\begin{equation} \label{Eq: expansion R cons}
	\R =\frac{\rho_{\left(0 \right)} \l}{2 \mathrm{\pi}} + \frac{q_{\left(0 \right)} \bpartial_\phi }{2 \mathrm{\pi}} + \sum_{\scriptsize\begin{array}{c} m =-\infty \\ m \neq 0 \end{array}}^\infty \frac{1}{\mathrm{\pi}}  \left(\rho_\mm \mathbfit{Y}_\mm + q_\mm \Psibf_\mm\right),
\end{equation}
with the coefficients derived from
\begin{definition}[Vector Fourier coefficients]\label{Vector Fourier coefficients} 
\begin{align}
	\rho_\mm &:= \int_{S^1} \R \cdot \Y_\mm \dex \phi = \int_{S^1} \rho Y_\mm \dex \phi \\
	q_\mm &:= \int_{S^1} \R \cdot \Psibf_\mm \dex \phi = \left\{\begin{array}{rl} \int_{S^1} q \, \dex \phi, & \mbox{if } m = 0 \\ \frac{1}{|m|} \int_{S^1} q \partial_\phi Y_\mm \dex \phi, & \mbox{if } m \neq 0. \end{array} \right. 
\end{align}
\end{definition}
Of course we find the according expansion of $\kappaII $
\begin{align} \label{Eq: expansion kappa cons}
	\kappaII &= \frac{\rho_{\left(0 \right)} \i_\l \dVS }{2 \mathrm{\pi}} + \frac{q_{\left(0 \right)} \i_{\bpartial_\phi} \dVS }{2 \mathrm{\pi}} + \sum_{\scriptsize\begin{array}{c} m =-\infty \\ m \neq 0 \end{array}}^\infty \frac{1}{\mathrm{\pi}}  \left(\rho_\mm \i_{\mathbfit{Y}_\mm} \dVS + q_\mm \i_{\Psibf_\mm} \dVS \right), \\
	&=: \frac{\mathbfit{\kappa}^\Y_{(0)} }{2 \mathrm{\pi}} + \frac{\mathbfit{\kappa}^\Psibf_{(0)} }{2 \mathrm{\pi}} + \sum_{\scriptsize\begin{array}{c} m =-\infty \\ m \neq 0 \end{array}}^\infty \frac{1}{\mathrm{\pi}} \left(\mathbfit{\kappa}^\Y_\mm + \mathbfit{\kappa}^\Psibf_\mm \right)
\end{align}
As introduced so far, the expansion in vector spherical harmonics seemingly does not yield a meaningful spatial expansion of the dislocation state. To arrive at the tensor expansion we make use of the equivalence of the expansion in spherical harmonics and likewise trigonometric functions with an expansion into traceless symmetric tensors of increasing order.
\subsection{Irreducible tensor expansion} \label{Sec: Irreducible tensor expansion}
The tensor expansion for the density function as known, e.g., from the theory of liquid crystals \citep{hess75,blenk_em91}, takes the form
\begin{align}
	\hat{\rhobf}_{(0)} &:= \int_{S} \rho (\l) \dS  \\
	\hat{\rhobf}_{(1)} &:= \int_{S} \rho (\l) \l \dS  \\
	\hat{\rhobf}_{(2)} &:= \int_{S} \rho (\l) \l \totimes \l   \\
	\hat{\rhobf}_{(n)} &:= \int_{S} \rho (\l) \underbrace{\l \totimes \cdots  \totimes \l}_{k\mathrm{-times}} \dS = \int_{S} \rho (\l) \l^{\totimes k} \dS\label{Eq: standard expansion, dislocation},
\end{align} 
where we denote with $\totimes$ a tensor product followed by a symmetrizing operation and an operation removing the trace from all contractions of two indices, such that the result is a fully symmetric traceless tensor. Furthermore, we introduced an exponential notation for the $n$-fold such product of a single vector with itself as seen in the right most expression of the last line. The density function for three-dimensional distributions may then be recovered from the following series expansion
\begin{equation} \label{Eq: tensor series}
	\rho(\l) =\frac{1}{4 \mathrm{\pi}} \left( \hat{\rho}_{(0)} + \sum_{k=1}^\infty  \frac{\left(2k+1\right)!!}{k!} \hat{\rho}_{i_1 \cdots  i_k} l^{i_1} \cdots l^{i_k} \right).
\end{equation}
The double factorial symbol is defined for odd numbers through $ \left(2k+1\right)!! := \left(2k+1\right)\left(2k-1\right)\cdots 3 \cdot 1 $. We note that the tensors in Eq.\ \eqref{Eq: standard expansion, dislocation} are defined contravariantly, that is with upper indices. The covariant tensor components in the series expression Eq.\ \eqref{Eq: tensor series} are obtained through lowering all indices with help of the standard metric $ g_{ij} = \delta_{ij} $. Such connection is always implied when we use a symbol with the same number of indices in different (upper or lower) positions.

The relation to spherical harmonics is hidden in the components of the symmetric traceless tensor $ \l^{\totimes k} $ \citep{applequist89}. The coefficients of the tensor $\l^{\totimes k} $ are either zero or linear combinations of the spherical harmonics $Y_\km$. In turn, every $Y_\km$ with $-k \leq m \leq k$ appears in $\l^{\totimes k} $. Because the general formulas for the coefficients are cumbersome we use the following short hand notation for the coefficients of the traceless tensors
\begin{equation}
	\l^{\totimes k} = \lhat^{i_1 \cdots i_k} \bpartial_{i_1} \otimes \cdots \otimes \bpartial_{i_k}.
\end{equation} 

The tensor expansion Eq.\ \eqref{Eq: standard expansion, dislocation} corresponds to the expansion of a scalar with spherical harmonics. From the definition of the vector spherical harmonics it is straight forward to define a `vectorial' tensor expansion using the base tensors
\begin{align}
	\hat{\mathbfit\Lambda}_{(k)} &:= \l \otimes \l^{\totimes k}  = l^i \lhat^{i_1 \cdots i_k} \bpartial_i \otimes \bpartial_{i_1} \otimes \cdots \otimes \bpartial_{i_k} \\
	\hat{\mathbfit\Gamma}_{(k)} &:= \frac{1}{\sqrt{k(k+1)}} \grad \l^{\totimes k} = \frac{1}{\sqrt{k(k+1)}} g^{\mu \nu} \partial_\nu \lhat^{i_1 \cdots i_k} \bpartial_\mu \otimes \bpartial_{i_1} \otimes \cdots \otimes \bpartial_{i_k} \\
	\hat{\mathbfit\Gamma}^*_{(k)} &:= \frac{1}{\sqrt{k(k+1)}} \grad^* \l^{\totimes k} = \frac{1}{\sqrt{k(k+1)}} \varepsilon^{\mu \nu} \partial_\nu \lhat^{i_1 \cdots i_k} \bpartial_\mu \otimes \bpartial_{i_1} \otimes \cdots \otimes \bpartial_{i_k}
\end{align} 
Note that the degree of the tensors with index $k$ each time is $k+1$, which may be confusing but conforms with the notation used in the case of vector spherical harmonics. With these tensors we obtain a series of tensors through
\begin{definition}[Irreducible tensor expansion]\label{Def: Irreducible tensor expansion}
\begin{align}
	\hat{\rhobf}_{(k)} &:= \int_{S^2} \R \cdot \hat{\mathbfit\Lambda}^{(k)} \dS \\
	\hat{\q}_{(k)} &:=  \int_{S^2} \R \cdot \hat{\mathbfit\Gamma}^{(k)} \dS  \\
	\hat{\q}^{*}_{(k)} &:= \int_{S^2} \R \cdot \hat{\mathbfit\Gamma}^*_{(k)}  \dS.
\end{align} 
\end{definition}
We remark that for the tensors on the right hand side, as opposed to the base tensors within the integrals, the index $k$ \emph{is} the degree of the tensor. For the higher dimensional vector field $ \R $ we now obtain the series representation
\begin{equation} 
	\R(\l) = \frac{1}{4\mathrm{\pi}} \left[ \rho_{0} \l + \sum_{k=1}^\infty  \frac{\left(2k+1\right)!!}{k!} \left( \hat\rho_{i_1 \cdots  i_k} 	\hat{\Lambda}^{i i_1 \cdots i_k} \bpartial_i + \hat{q}_{i_1 \cdots  i_k} \hat{\Gamma}^{\mu i_1 \cdots i_k } \bpartial_\mu  + \hat{q}^{*}_{i_1 \cdots  i_k} \hat{\Gamma}^{\mu *i_1 \cdots i_k } \bpartial_\mu \right) \right].
\end{equation}

In the case of pure glide the tensor expansion Eq.\ \eqref{Eq: standard expansion, dislocation} remains unchanged but the series representation takes the form \citep{zheng_z01}\footnote{In \cite{hochrainer13b} I mistakenly assumed Eq.\ \eqref{Eq: standard expansion, dislocation} to be valid for the 2D case by just changing  $ 4 \mathrm{\pi} $ to $ 2 \mathrm{\pi} $. This remained unrecognized because it had no consequences for the rest of that paper.}
\begin{equation} \label{Eq: tensor series pure glide}
	\rho(\l) =\frac{1}{2 \mathrm{\pi}} \left( \hat{\rho}_{(0)} + \sum_{k=1}^\infty  2^k \hat{\rho}_{i_1 \cdots  i_k} l^{i_1} \cdots l^{i_k} \right).
\end{equation}
For the quasi two-dimensional case we define the traceless base tensors as
\begin{align}
	\hat{\mathbfit\Lambda}_{(k)} &:= l^i \lhat^{i_1 \cdots i_k} \bpartial_i \otimes \bpartial_{i_1} \otimes \cdots \otimes \bpartial_{i_k} \\
  \hat{\mathbfit\Gamma}_{(0)} &:= \bpartial_\phi \\
	\hat{\mathbfit\Gamma}_{(k)} &:= \frac{1}{k} \partial_\phi \lhat^{i_1 \cdots i_k} \bpartial_\phi \otimes \bpartial_{i_1} \otimes \cdots \otimes \bpartial_{i_k}
\end{align} 
and obtain the expansion
\begin{equation} 
	\R(\l) = \frac{1}{2\mathrm{\pi}} \left[ \rho_{0} \l + q_{0} \bpartial_\phi +\sum_{k=1}^\infty  2^k \left( \hat\rho_{i_1 \cdots  i_k} 	\hat{\Lambda}^{i i_1 \cdots i_k} \bpartial_i + \hat{q}_{i_1 \cdots  i_k} \hat{\Gamma}^{\varphi i_1 \cdots i_k} \bpartial_\phi  \right) \right],
\end{equation}
with the coefficients given by
\begin{definition}[Irreducible tensor expansion in pure glide]\label{Irreducible tensor expansion (pure glide)}
\begin{align}
	\hat{\rhobf}_{(k)}  &:= \int_{S^1} \R \cdot \hat{\mathbfit\Lambda}^{(k)} \dex \phi \\
	\hat{\q}_{(k)}  &:=  \int_{S^1} \R \cdot \hat{\mathbfit\Gamma}^{(k)} \dex \phi. 
\end{align} 
\end{definition}
\subsection{Reducible tensor expansion} \label{Sec: Reducible tensor expansion}
Besides the irreducible expansion discussed in the last subsection we also introduce a reducible expansion. The reducible tensors are more easily interpreted geometrically than the irreducible ones, similar to the stress tensor being more easily visualized than the stress deviator.  For the reducible expansion we introduce the base tensors with non-vanishing traces
\begin{align}
	\mathbfit\Lambda_{(k)} &:= \l \otimes \l^{\otimes k} = l^i l^{i_1} \cdots l^{i_k} \bpartial_i \otimes \bpartial_{i_1} \otimes \cdots \otimes \bpartial_{i_k} \\
	\mathbfit\Gamma_{(k)} &:= \frac{1}{\sqrt{k(k+1)}} \grad \l^{\otimes k} = \frac{1}{\sqrt{k(k+1)}} g^{\mu \nu} \partial_\nu \left( l^{i_1} \cdots l^{i_k} \right) \bpartial_\mu \otimes \bpartial_{i_1} \otimes \cdots \otimes \bpartial_{i_k}  \\ &= \frac{k}{\sqrt{k(k+1)}} g^{\mu \nu} \bpartial_\mu \otimes \e_\nu \sotimes \l^{\otimes (k-1)}   \\
	\mathbfit\Gamma^*_{(k)} &:= \frac{1}{\sqrt{k(k+1)}} \grad^* \l^{\otimes k} = \frac{1}{\sqrt{k(k+1)}} \varepsilon^{\mu \nu} \partial_\nu \left( l^{i_1} \cdots l^{i_k} \right) \bpartial_\mu \otimes \bpartial_{i_1} \otimes \cdots \otimes \bpartial_{i_k} \\ &= \frac{k}{\sqrt{k(k+1)}} \varepsilon^{\mu \nu} \bpartial_\mu \otimes \e_\nu \sotimes \l^{\otimes (k-1)},
\end{align} 
where $\sotimes$ denotes the tensor product followed by a complete symmetrization of the resulting tensor (without removing the trace). Accordingly, we define reducible series of spatial tensors in
\begin{definition}[Reducible tensor expansion]\label{Def: Reducible tensor expansion}
\begin{align}
	\rhobf_{(k)} := \int_{S^2} \R \cdot \mathbfit\Lambda_{(k)} \dS &= \int_{S^2} \rho \l^{\otimes k} \dS \\
	\q_{(k)} :=  \int_{S^2} \R \cdot \mathbfit\Gamma_{(k)} \dS  &= \frac{k}{\sqrt{k(k+1)}}  \int_{S^2} q^\nu \e_\nu \sotimes \l^{\otimes (k-1)}  \\
	\q^{*}_{(k)} := \int_{S^2} \R \cdot \mathbfit\Gamma^*_{(k)} \dS &=\frac{k}{\sqrt{k(k+1)}} \int_{S^2} q^\mu \tensor{\varepsilon}{_\mu^\nu} \e_\nu \sotimes \l^{\otimes (k-1)}.
\end{align} 
\end{definition}
These reducible tensors `contain' each other, such that the tensors of different order are not independent of each other. The relation between the tensors is found by the trace operation as
\begin{align}
	\tr \rhobf_{(k)}&= \rhobf_{(k-2)} \\
	\tr \q_{(k)}  &= \sqrt{\frac{(k-2)(k-1)}{k(k+1)}}  \q_{(k-2)} \\
	\tr \q^*_{(k)}  &= \sqrt{\frac{(k-2)(k-1)}{k(k+1)}}  \q^*_{(k-2)},
\end{align} 
which is true for $ k \geq 2 $. The irreducible series maybe obtained from the reducible one by a detracer operation introduced in general by \citet{applequist89}. 

The reducible basis tensors for the case of pure glide are given as
\begin{align}
	\mathbfit\Lambda_{(k)} &:= \l \otimes \l^{\otimes k} \\
  \mathbfit\Gamma_{(0)} &:= \bpartial_\phi \\
	\mathbfit\Gamma_{(k)} &:= \frac{1}{k} \partial_\phi \left( l^{i_1} \cdots l^{i_k} \right) \bpartial_\phi \otimes \bpartial_{i_1} \otimes \cdots \otimes \bpartial_{i_k} = \bpartial_\phi \otimes \e_\phi \sotimes \l^{\otimes (k-1)}.
\end{align} 
The coefficients are provided in 
\begin{definition}[Reducible tensor expansion in pure glide]\label{Def: Reducible tensor expansion (pure glide)}
\begin{align}
	\rhobf_{(k)}  &= \int_{S^1} \R \cdot \mathbfit\Lambda^{(k)} \dex \phi \\
	\q_{(k)}  &=  \int_{S^1} \R \cdot \mathbfit\Gamma^{(k)} \dex \phi.   
\end{align}
\end{definition}
The relation between the tensors and their traces reads in this case
\begin{align}
	\tr \rhobf_{(k)}&:= \rhobf_{(k-2)} \\
	\tr \q_{(k)}  &:= \frac{(k-2)}{k}  \q_{(k-2)}
\end{align} 
which is again valid for $ k \geq 2 $.
\subsection{Reducible expansion in tensor valued differential forms} \label{Reducible expansion in vvdf}
Finally we define another series of tensors, which are essentially equivalent to the reducible tenors introduced for the density $ \rho $ above. These tensor valued spatial two-forms are provided in
\begin{definition}[Reducible differential form expansion]\label{Reducible differential form expansion}
\begin{equation}
	\mathbfit{K}_{(k)} := \int_{S^{2\left(1\right)}} \kappaII \otimes \l^{\otimes k}. 
\end{equation} 
\end{definition}
We note that in this definition there is no formal difference between the general case and the case of pure glide, which are only distinguished by the domain of integration being $ S^2 $ or $ S^1 $. As is easily checked, the tensors are related to the reducible tensors $\rhobf_{(k)} $ through
\begin{equation}
	\tensor{K}{_{ij} ^{i_1 \cdots i_k}} = \varepsilon_{ijl} \rho^{l i_1 \cdots i_k} \text{  and  }  \rho^{i_1 \cdots i_k} = \frac{1}{2} \varepsilon^{i_1 ij} \tensor{K}{_{ij} ^{i_2 \cdots i_k}}.
\end{equation} 
Although it may seem idle to introduce this extra series it will turn out to be beneficial in closing the hierarchy of evolution equations derived in Section \ref{Sec: Evolution alignment}, as discussed in Section \ref{Sec: Preliminaries}. Furthermore, these tensors come with a different mathematical flavor. As differential two-forms these tensors may be meaningfully integrated over surfaces, while the tensors $\rhobf_{(k)} $ are more naturally integrable over volumes. We will also find seemingly different evolution laws for these series which reflect the natural conservation law for two- or three-forms. It should be noted that there is no scalar element in the series $\mathbfit{K}_{(0)}$ and the zeroth element corresponds to $ \rhobf_{(1)}$. However, the zeroth order term $ \rho_{(0)} $ will still be obtainable by a trace operation. In fact we find a direct relation between the two series by performing a specific trace operation on the tensors $\mathbfit{K}_{(k)} $. This trace operation is defined for tensor valued two-forms such that it turns a tensor valued two-form with $k$ contra-variant indices into a tensor-valued 3-form (volume form) with $k-1$ contra-variant indices through
\begin{equation}
	\tr \mathbfit{K}_{(k)} = \tr \left( \tensor{K}{_{ij} ^{m_1 \cdots m_k} } \dx{i} \wedge \dx{j} \otimes \bpartial_{m_1} \otimes \cdots \otimes \bpartial_{m_k} \right) = g_{n m_1} \tensor{K}{_{ij} ^{m_1 m_2 \cdots m_k} } \dx{i} \wedge \dx{j} \wedge \dx{n} \otimes \bpartial_{m_2} \otimes \cdots \otimes \bpartial_{m_k}.
\end{equation} 
On a fully decomposable vector valued two-form $ \mathbfit{\omega} = \omega \otimes \u_{(1)} \otimes \cdots \otimes \u_{(k)} $ the trace yields
\begin{equation}
	\tr \mathbfit{\omega} = \omega \wedge \u^\flat_{(1)} \otimes \u_{(2)} \otimes \cdots \otimes \u_{(k)},
\end{equation} 
where the symbol $\left(\cdot\right)^\flat$ denotes the operation of lowering indices, i.e., for turning a vector into a one-form,
\begin{equation}
	\u^\flat = \left(u^i \bpartial_i \right)^\flat = g_{ij} u^j \dx{i}.
\end{equation}
This directly yields that the traces of the tensors $ \mathbfit{K}_{(k)} $ are up to the volume form part $\dV $ the tensors $\rhobf_{(k)} $
\begin{equation}
	\tr \mathbfit{K}_{(k)} = \int_{S^{2\left(1\right)}} \kappaII \wedge \l^\flat \otimes \l^{\otimes k-1} \dS = \left( \int_{S^{2\left(1\right)}} \rho \l^{\otimes k-1} \dS \right) \dV = \rhobf_{(k-1)} \dV.
\end{equation} 
Because of this relation we interpret $\rhobf_{(k)} $ as being differential three-forms `by nature'. Most obviously this makes sense in the light of the scalar $ \rho_0 $, which gives the line length of dislocations per unit volume; a quantity which clearly may be integrated over volumes. On the other hand the zeroth element of the two-form series is, $ \mathbfit{K}_{(0)} = \kappabf $, i.e., the classical dislocation density two-form which may be meaningfully integrated over surfaces.
\subsection{Interpretation of the alignment tensors} \label{Sec: Interpretation}
As just seen the zeroth terms in the two reducible tensor expansions are well known dislocation density measures. The higher order tensors are geometrically interpreted as follows: the tensors $ \mathbfit{K}_{(1)} $ (or likewise $ \rhobf_{(2)} $) assigns to a surface element spanned by two vectors $ \u $ and $ \mathbfit{w} $ (or defined by its normal $ \n $) the average non-oriented (!) tangent direction of the dislocations intersecting the surface. Diagonal components of $ \rhobf_{(2)} $ assign the density of dislocations parallel to the normal while off-diagonal elements account for dislocations parallel to the surface. As a symmetric tensor $ \rhobf_{(2)} $ may be visualized by an ellipsoid, the average radius of which defines the total dislocation density, while the ellipticity, i.e., the deviation of the dislocation direction distribution from being isotropic, is defined by its traceless part (or deviator) $ \hat\rhobf_{(2)} $. The next higher order tensors $ \mathbfit{K}_{(2)} $ and $ \rhobf_{(3)} $ assign to a surface element the average tensors projecting vectors on their component in the direction of the dislocations. In analogy to the ellipsoid of the second order tensor being shaped by three diametrically opposed symmetric extrema of each time the same hight, the traceless third order tensor introduces on a deformed sphere diametrically opposed antisymmetric extrema where local maxima are opposed to local minima \citep{hochrainer13b}. This is also how the tensors generalize to higher orders: the traceless tensors of even rank characterize symmetric extrema (directions of high or low total dislocation density) while those of odd rank deal with an increasing number of antisymmetric extrema, that is, with directions of high net dislocation density.

In the case of glide only, some components of the second order tensor $ \rhobf_{(2)} $ already appeared in theories distinguishing densities of screw and edge dislocation, $\rho_\mathrm{s}$ and $ \rho_\mathrm{e} $, respectively. Within the coordinate system introduced for the pure glide case, such that the glide direction $\bparx = \m $ and $ \bpary = \n \times \m $ these densities correspond to the diagonal elements, i.e.,
\begin{align}
	\rho_\mathrm{s} &= \rho_{11} \\
	\rho_\mathrm{e} &= \rho_{22}.
\end{align} 
As to the authors knowledge, the off-diagonal element $ \rho_{12} $ was not yet introduced to the literature. This component is obviously needed to obtain the density of dislocations with an arbitrary line direction within the glide plane. Moreover, the off-diagonal element is necessary to be able to represent arbitrary orientations of the ellipse connected to the quasi-two-dimensional tensor as without it the eigendirections are bound to be aligned with the edge and screw directions.
\subsubsection{The curvature tensors} \label{Sec: Curvature tensors}
The meaning of the tensors $ \mathbfit{K}_{(k)} $ or $ \rhobf_{(k)} $ is the standard interpretation of alignment tensors as known from other disciplines dealing with distributions of line like objects -- even though odd tensors are barely considered there. However, the tensors characterizing the curvature, $\q_{(k)}$ and $\q^*_{(k)}$, appear to be uncommon. Their interpretation is not that well understood so far, at least not in the three-dimensional case. Before we try to make sense out of the tensors, we first report a striking observation, namely that the solenoidality of the SODT, i.e., the fact that dislocation do not end inside a crystal, is reflected in the remarkable relation
\begin{align}
	\dex \mathbfit{K}_{(k)} &= \int_{S^2} \dex \left( \kappaII \otimes \l^{\otimes k} \right)  \nonumber\\
	&= \int_{S^2} \left( \dex \kappaII \otimes \l^{\otimes k} + \kappaII \wedge \nabla \l^{\otimes k} \right) \nonumber \\
	&= \left( \int_{S^2}  \R \cdot \nabla \l^{\otimes k} \dS \right) \dV \nonumber \\
	&= \sqrt{k(k+1)} \q_{(k)}  \dV.
\end{align} 
This directly translates to
\begin{equation}
	\divergence  \rhobf_{(k)} = \sqrt{(k-1)k}  \q_{(k-1)},
\end{equation} 
which may alternatively be derived from Eq.\ \eqref{Eq: d kappaII coord} using the product rule of the divergence operator and Gauss integration theorem. In the two-dimensional case we find
\begin{equation}
	\dex \mathbfit{K}_{(k)}  = k \q_{(k)}  \dV \text{ and  } \divergence  \rhobf_{(k)} = (k-1) \q_{(k-1)},
\end{equation} 
respectively. The prefactors each time stem from the normalizing factors introduced in the definition of the base vectors.

That the curvature tensors appear as the divergence of the dislocation density tensors makes sense in light of the curvature being the change of the line direction along the dislocation line. The divergence operation contracts the derivative with the directional information contained in the density tensor. 

The vector $ \q_{(1)} $ yields the average (or net) curvature vector similar to the net dislocation density vector $ \kappavec $, as 
\begin{equation}
	\q_{(1)} =  \frac{1}{\sqrt{2}}  \int_{S^2} q^\nu \e_\nu,
\end{equation} 
where we note that $ q^\nu \e_\nu $ is the spatial representation of the curvature vector. The second order curvature tensor is a symmetric traceless tensor
\begin{equation}
	\q_{(2)} =  \frac{1}{\sqrt{6}}  \int_{S^2} q^\nu \e_\nu \sotimes \l,
\end{equation} 
which assigns to a volume the average projection tensor composed of line direction and curvature vector. The higher order curvature tensors $\q_{(k)}$ contain information on higher frequency fluctuations of the curvature vector $ \q $ on the sphere. Seemingly the series $\q_{(k)}$ contains all information on the curvature state. Moreover the series is obtained from the divergence of the tensors $ \rhobf_{(k)} $ which renders the information contained in the series $\q_{(k)}$ redundant. So the question arises, whether the vectorial nature of $ \kappaII $, or equivalently $ \R $, needs a separate representation in terms of curvature tensors, as seemingly all information is already contained in the alignment tensors $ \rhobf_{(k)} $. This question may be restated as to what additional information is contained in the series $\q^*_{(k)}$ in the general, and in the coefficient $ q_{(0)} $ in the pure glide case. We first take a look at the two-dimensional case of pure glide. In that case the zeroth coefficient $ q_{(0)} $ is not derivable as the divergence of the first order tensor $ \rhobf_{(1)} = \kappavec $, as the latter is solenoidal. The zeroth order term is given by $ q_{(0)} = \int_{S^1} q \dex \phi $ and represents the constant part of the curvature distribution. In the general case it takes the whole series $\q^*_{(k)}$ to be able to reconstruct $ \q $ on the configuration space. To understand this series we take an exemplary look at the first term of it, which is defined as
\begin{equation}
	\q^*_{(1)} =  \frac{1}{\sqrt{2}}  \int_{S^2} q^\mu \tensor{\varepsilon}{_\mu^\nu} \e_\nu.
\end{equation} 
Notably, $ q^\mu \tensor{\varepsilon}{_\mu^\nu} \e_\nu $ is the curvature vector $ q^\nu \e_\nu $ tilted by $90^{\circ}$ within the tangent plane to the sphere. This may be rationalized by realizing that $  q^\mu \tensor{\varepsilon}{_\mu^\nu} \e_\nu = \l \times q^\nu \e_\nu $, where $ \times $ denotes the cross product. Due to this rotation $\q^*_{(1)}$ contains comparable information about the constant part of the curvature as found in $ q_{(0)} $ in the pure glide case. The higher order tensors $\q^*_{(k)}$ contain similar information as the $ \q_{(k)} $, only that the averages are obtained for the curvature vector tilted by $90^{\circ}$.

We summarize that we obtained three series of traceless or non-traceless tensors which together contain the same information on the local dislocation state as the second order dislocation density tensor, which is defined on a  higher dimensional configuration state. In light of this expansion, the two well known dislocation density measures,the scalar total dislocation density $\rho_0 $ and the dislocation density vector $\kappavec$ are only the first two elements of the series of alignment tensors. The second order tensor $ \rhobf_{(2)} $ contains information also found in screw--edge approaches to dislocation densities on a slip system level. The series of divergence tensors of the alignment tensors turn out to be one part of the series expansion of the dislocation curvature $ \q $ in the SODT. This series is seconded by another series of curvature tensors $ \q^*_{(k)}$ and only both series together allow the reconstruction of the SODT. In the two-dimensional case of glide only the second series reduces to a single scalar quantity, $ q_0$. The curvature tensors seem to be of limited value for a static description of dislocations in terms of tensors of increasing order. But, as we will see in the next Section, the curvature information is essential for the evolution of the alignment tensors.

\section{Expansion of the evolution equations} \label{Sec: Evolution alignment}

In the current section we provide the evolution equations for the coefficients of the second order dislocation density tensor in the expansions in terms of vector spherical harmonics, Fourier coefficients, and in terms of the various tensor coefficients introduced in the last Section. The derivation is provided in detail for the evolution of the coefficients in the expansion into vector spherical harmonics. The evolution equations for the Fourier coefficients and the tensor coefficients do then mostly follow by analogy. In each case we obtain infinite hierarchies of evolution equations. The evolution equations of the reducible alignment tensor will be used in Section \ref{Sec: Closure approximations} to entertain geometrically motivated closure assumptions needed to terminate the infinite hierarchy at low order. The basis for the derivations is evolution equation Eq.\ \eqref{Eq: dt kappaII}.

\subsection{Evolution in terms of vector spherical harmonics and Fourier coefficients} \label{Sec: Evolution VSH Fourier}

The evolution equation Eq.\ \eqref{Eq: dt kappaII} of the SODT can be used to derive evolution equations for the coefficients in the expansion via vector spherical harmonics which we summarize in
\begin{proposition}\label{Evolution VSH} The evolution of the coefficients in the expansion of the SODT with vector spherical harmonics introduced in \emph{Definition \ref{VSH coefficients}} is given by
\begin{align} \label{Eq: evolution coefficients VSH}
	\partial_t \rho_\km  &=  - \partial_i \int_{S^2} Y_\km \rho v^i \dS + \int_{S^2} \left[ \rho \vartheta^\mu \partial_\mu Y_\km -  g_{\mu \nu} v^\mu q^\nu Y_\km \right] \dS \\
	\partial_t q_\km &= - \frac{1}{\sqrt{k \left(k+1 \right)}} \partial_i \int_{S^2} \left[ q^\mu \partial_\mu Y_\km v^i - \vartheta^\mu \partial_\mu Y_\km \rho l^i \right] \dS \\
	\partial_t q^*_\km &= - \frac{1}{\sqrt{k \left(k+1 \right)}} \partial_i \int_{S^2} \left[ q^\mu \tensor{\varepsilon}{_\mu^\nu} \partial_\nu Y_\km v^i - \vartheta^\mu \tensor{\varepsilon}{_\mu^\nu} \partial_\nu Y_\km \rho l^i \right] \dS - \nonumber \\ &{} \qquad \qquad \qquad \qquad \qquad \qquad \qquad \qquad \qquad \qquad \qquad \sqrt{k\left(k+1\right)} \int_{S^2} Y_\km \varepsilon_{\mu \nu} \vartheta^\mu q^\nu \dS.  
\end{align}
\end{proposition}
\begin{proof}
We start by introducing a covariant version of the vector spherical harmonics as differential one forms, i.e.,
\begin{align}
	\Y_\km^\flat &= Y_\km \l^\flat = Y_\km \delta_{ij} l^j \dex x^i \\
	\Psibf_\km^\flat &= \frac{1}{\sqrt{k\left(k+1\right)} } \dex Y_\km =\frac{1}{\sqrt{k\left(k+1\right)} } \partial_\mu Y_\km \dex \mu \\
	\Phibf_\km^\flat &= \frac{1}{\sqrt{k\left(k+1\right)} } \dex^* \left(Y_\km \dS \right) = \frac{1}{\sqrt{k\left(k+1\right)} } \tensor{\varepsilon}{_\mu^\nu} \partial_\nu Y_\km \dex \mu,
\end{align}
where $ \left( \cdot \right) ^\flat $ denotes the operation of lowering indices and $ \dex^* $ is the so-called interior derivative, that is, the $L^2$ adjoint operator to the exterior derivative $ \dex$. While the latter's definition in general requires using the Hodge-star operator, the reader not familiar with this terminology may take the above coordinate expression as the definition. With this definition the integrals for obtaining the coefficients may be reformulated in differential form formalism. The integrals each time turn into integrals of a five-form over the two-dimensional directional space leaving a volume form on the base space as result: 
\begin{align}
	\rho_\km dV &= \int_{S^2}  \Y_\km^\flat \wedge \kappaII \\
	q_\km dV &= \int_{S^2} \Psibf_\km^\flat  \wedge \kappaII \\
	q^*_\km dV &= \int_{S^2} \Phibf_\km^\flat \wedge \kappaII
\end{align}
The relations to those of Definition \ref{VSH coefficients} follow from the product rule of the interior multiplication $ \i $ which yields for a 1-form $ \H^\flat = H_i \dex x^i + H_\mu \dex \mu $ obtained from a vector field $ \H = H^i \bpartial^i + H^\mu \bpartial_\mu $ on the configuration space 
\begin{equation} 
  \H^\flat \wedge \kappaII = \H \cdot \R \dVS.
\end{equation}

In preparation for the derivation of the evolution equations of the coefficients we collect a few formulas which are obtained from the product rules of interior multiplication $ \i $ and exterior derivative $ \dex $ for a 1-form $ \H^\flat $ as above. The time derivative of the wedge product of $ \H $ and $ \kappaII $ is easily found as  
\begin{equation} \label{Eq: kappa product evol general}
	\partial_t \int_{S^2} \H^\flat  \wedge \kappaII = \dex \int_{S^2} \left( \H^\flat  \wedge \i_\V \kappaII \right) - \int_{S^2} \dex \H^\flat  \wedge \i_\V \kappaII.
\end{equation}
We furthermore take advantage of the following two identities
\begin{align} 
	\H^\flat  \wedge \i_\V \kappaII &= - \i_\V \left( \H^\flat  \wedge \kappaII \right) + \i_\V \H^\flat \wedge \kappaII \text{   and} \\
	\dex \H^\flat \wedge \i_\V \kappaII &= - \left( \i_\R \i_\V \dex \H^\flat \right) \dVS. 
\end{align}
For the exterior derivatives of the vector spherical harmonics base forms we find
\begin{align}
	\dex \Y_\km^\flat &= \dex Y_\km \wedge \l^\flat + Y_\km \dex \l^\flat = \sqrt{k\left(k+1\right)} \Psibf_\km^\flat \wedge \l^\flat + Y_\km \dex \l^\flat, \\
	\dex \Psibf_\km^\flat &= \frac{1}{\sqrt{k\left(k+1\right)} } \dex \dex Y_\km = 0, \\
	\dex \Phibf_\km^\flat &= \frac{1}{\sqrt{k\left(k+1\right)} } \dex \dex^* \left(Y_\km \dS \right) = - \sqrt{k\left(k+1\right)}  Y_\km \dS.
\end{align}
The last identity follows from the fact that $ Y_\km $ are eigenfunction of the Laplace-Beltrami-operator $ \dex^* \dex$ to the eigenvalue $ - k\left(k+1\right)$, which makes $ Y_\km \dS $ an eigenform to the same eigenvalue of the Hodge-de Rham-Laplace-operator $ \dex \dex^* $ on differential two-forms. Because of the product structure of the sphere bundle $ M \otimes S^2 $ this relation remains true on the configuration space.

We make a few observations concerning the terms appearing in the following calculations. For their derivation we employ (additional to the product rules used before) that for two vector fields $ \mathbfit{G} $ and $ \H $ on the configuration space we have $ \i_\mathbfit{G} \H^\flat = \mathbfit{G} \cdot \H $. We obtain the following identities
\begin{align}
	\i_\V \Y_\km^\flat &= 0 \\
	\i_\V \left(  \dex Y_\km \wedge \l^\flat \right) &= \vartheta^\mu \partial_\mu Y_\km \l^\flat  \\
	\i_\V Y_\km \dex \l^\flat &= Y_\km \left( - v^i \partial_\mu l_i \dex \mu + \vartheta^\mu \partial_\mu l_i \dex x^i \right).
\end{align}
In order to understand the geometric meaning of the last line, we calculate
\begin{equation}
	\partial_\mu l_i =\partial_\mu g_{ij} l^j = g_{ij} \partial_\mu l^j = g_{ij} \dex x^j (\e_\mu) = \g\left( \bpartial_i , \e_\mu \right).
\end{equation}
Accordingly we find
\begin{align}
	- v^i \partial_\mu l_i \dex \mu &=  - \g \left( \v, \e_\mu \right) \dex \mu, \\
	\vartheta^\mu \partial_\mu l_i \dex x^i &= \g\left( \bpartial_i , \thetabf_\mathrm{h} \right) \dex x^i,
\end{align}
such that
\begin{equation}
	\i_\V \dex \l^\flat = - \g \left( \v, \e_\mu \right) \dex \mu + \g\left( \bpartial_i , \thetabf_\mathrm{h} \right) \dex x^i.
\end{equation}
As a last preparatory calculation we provide the direct consequence
\begin{align}
	\i_\R \i_\V \dex \l^\flat &= - \g \left( \v, q^\mu \e_\mu \right) + \g\left( \rho l^i \bpartial_i , \thetabf_\mathrm{h} \right) = - \g \left( \v, \q_\mathrm{h} \right) \nonumber \\ &= -  g_{\mu \nu} v^\mu q^\nu,
\end{align}
where we used that $ \thetabf_\mathrm{h} $ is orthogonal to $ \l $. The result is the scalar product of the velocity vector and the curvature vector, which characterizes the rate of line-length increase (or shrinkage) due to bowing out of curved dislocations.

With these preparations we find the evolution equation for the normal coefficients
\begin{align}
	\partial_t \rho_\km dV &= \dex \int_{S^2} \Y_\km^\flat \wedge \i_\V \kappaII  - \int_{S^2} \dex \Y_\km^\flat \wedge \i_\V \kappaII \nonumber  \\
	&= \dex \int_{S^2} - \i_\V \left( \Y_\km^\flat \wedge \kappaII \right) + \int_{S^2} \left[ \i_\R \i_\V  \left( \dex Y_\km \wedge \l^\flat \right)  + \i_\R \i_\V Y_\km \dex \l^\flat \right] \dVS \nonumber \\
	&= - \dex \int_{S^2} Y_\km \rho \i_\V \dVS + \int_{S^2} \left[ \rho \vartheta^\mu \partial_\mu Y_\km -  g_{\mu \nu} v^\mu q^\nu Y_\km \right] \dVS
\end{align}
Similarly we find for the first tangential coefficient-forms
\begin{align}
	\partial_t q_\km dV &= \dex \int_{S^2} \Psibf_\km^\flat \wedge \i_\V \kappaII \nonumber \\
	&= \dex \int_{S^2} \left[ - \i_\V \left( \Psibf_\km^\flat \wedge \kappaII \right) + \i_\V \Psibf_\km^\flat \wedge \kappaII  \right] \nonumber \\
	&= - \frac{1}{\sqrt{k \left(k+1 \right)}} \dex \int_{S^2} \left[ q^\mu \partial_\mu Y_\km \i_\V \dVS - \vartheta^\mu \partial_\mu Y_\km \i_\R \dVS \right],
\end{align}
and the second tangential coefficient-forms
\begin{align}
	\partial_t q^*_\km dV &= \dex \int_{S^2} \Phibf_\km^\flat \wedge \i_\V \kappaII - \int_{S^2} \dex \Phibf_\km^\flat \wedge \i_\V \kappaII  \nonumber \\
	&= \dex \int_{S^2} \left[ - \i_\V \left( \Phibf_\km^\flat \wedge \kappaII \right) + \i_\V \Phibf_\km^\flat \wedge \kappaII \right] - \int_{S^2} -\sqrt{k\left(k+1\right)} Y_\km \dS \wedge \i_\V \kappaII \nonumber \\
	&= \dex \int_{S^2} \left[ - \i_\V \left( \Phibf_\km^\flat \wedge \kappaII \right) + \i_\V \Phibf_\km^\flat \wedge \kappaII \right] - \int_{S^2} \sqrt{k\left(k+1\right)} Y_\km \left( \i_\R \i_\V \dS \right) \nonumber \dVS \\
	&= - \frac{1}{\sqrt{k \left(k+1 \right)}} \dex \int_{S^2} \left[ q^\mu \tensor{\varepsilon}{_\mu^\nu} \partial_\nu Y_\km \i_\V \dVS - \vartheta^\mu \tensor{\varepsilon}{_\mu^\nu} \partial_\nu Y_\km \i_\R \dVS \right] - \nonumber \\ &{} \qquad \qquad \qquad \qquad \qquad \qquad \qquad \qquad \qquad \qquad \sqrt{k\left(k+1\right)} \int_{S^2} Y_\km \varepsilon_{\mu \nu} \vartheta^\mu q^\nu \dVS.
\end{align}
The results of the last three calculations can be translated into the evolution of the scalar coefficients as given in the proposition by considering the definition of the divergence $ \divergence $ of a vector field $ X $ through $ \divergence X \dV = \dex \left( \i_X \dV \right)$ (see, e.g., \citet{marsden_h83}).
\end{proof}
We additionally introduce a symbolic reformulation of Proposition \ref{Evolution VSH} which reads
\begin{align} \label{Eq: evolution spherical symbol}
	\partial_t \rho_\km  &=  - \divergence \int_{S^2} \R  \cdot \Y_\km \, \v \dS + \int_{S^2} \left[ \rho \, \V \cdot \nabla Y_\km -  \v \cdot \q \, Y_\km \right] \dS \\
	\partial_t q_\km &= - \divergence \int_{S^2} \left[ \R \cdot \Psibf_\km \v - \V \cdot \Psibf_\km  \rho \l \right] \dS \\
	\partial_t q^*_\km &= - \divergence \int_{S^2} \left[ \R \cdot \Phibf_\km \v - \V \cdot \Phibf_\km  \rho \l \right] \dS - \sqrt{k\left(k+1\right)} \int_{S^2} \thetabf \cdot \q^\perp Y_\km \, \dS,
\end{align}
where we introduced $ \q^\perp := \tensor{\varepsilon}{^\mu_\nu} q^\nu \bpartial_\mu $. Before we regard the evolution of the Fourier coefficients we provide a brief interpretation of the evolution equations derived above.

The time rate of change of the coefficients each time contains a divergence part and a part which would be looked at as source term. Notably, for the coefficients $ q_\km $ the source term is missing, such that all these coefficients define conserved quantities.

In the evolution equation of $ \rho_\km $ the divergence term accounts for the spatial flux of density while the term containing the rotational velocity $ \thetabf $ accounts for fluxes on the orientation space. Although this flux is also conservative globally it introduces a transfer of density `between different coefficients' and therefore appears as source term in the individual equations. The last term contains the scalar product between the spatial velocity vector $ \v $ and the curvature density vector $ \q $ and accounts for line length changes due to the bow out of dislocations.

In the evolution equation of $ q_\km $ the divergence term contains one term accounting for `curvature fluxes' and one term dealing with curvature changes due to rotations. The divergence term in the evolution of $ q^*_\km $ contains the according terms with a $90^\circ$ rotation. The source term in the $ q^*_\km $ evolution contains an antisymmetric product of the vectors $ \thetabf $ and $ \q $. Similar to the cross product in three dimensions this measures the area spanned by the two vectors. This term will for example be non trivial if a planar curve would be deformed into a non-planar one (out of plane rotation). This will be of importance when the geometry of cutting processes of dislocation are considered. In this case jogs may form by cutting polarized forests while moving or through the motion of other dislocations even for otherwise stationary dislocations \citep{arsenlis_etal04, hochrainer13}.

The evolution equations for the Fourier coefficients in the case of conservative dislocation motion are the topic of
\begin{proposition}\label{Evolution Vector Fourier} The evolution of the Fourier coefficients of the SODT introduced in \emph{Definition \ref{Vector Fourier coefficients}} is given by
\begin{align} \label{Eq: evolution coefficients Fourier}
	\partial_t \rho_\mm  &=  - \partial_i \int_{S^1} \rho v^i Y_\mm  \dS + \int_{S^1} \left[ \rho \vartheta \partial_\phi Y_\mm +  v q  Y_\mm \right] \dS \\
	\partial_t q_0 &= - \partial_i \int_{S^1} \left[ q v^i - \vartheta \rho l^i \right] \dS \\
	\partial_t q_\mm &= - \frac{1}{|m|} \partial_i \int_{S^1} \left[ q v^i \partial_\phi Y_\mm  - \vartheta \rho l^i \partial_\phi Y_\mm \right] \dS
\end{align}
\end{proposition}
\begin{proof}
The abstract derivation in terms of differential forms makes the derivation of the evolution equations in the case of glide only largely analogous to the treatment of the general case. The originally defined base fields are turned into differential forms as 
\begin{align}
	\Y_\mm^\flat &= Y_\mm \l^\flat = Y_\mm l_i \dex x^i, \\
	\Psibf_\mm^\flat &=  \left\{\begin{array}{rl} \dex \phi, & \mbox{if } m = 0 \\ \frac{1}{|m|} \partial_\phi Y_\mm \dex \phi, & \mbox{if } m \neq 0. \end{array} \right. 
\end{align}
The derivation remain \emph{mutas mutandis} the same as in the proof of Proposition \ref{Evolution VSH}. The evolution equations for the Fourier coefficients of the density is obtained as
\begin{equation}
	\partial_t \rho_{(k)} dV = - \dex \int_{S^1} Y_{(k)} \rho \i_\V \dVS + \int_{S^1} \left[ \rho \vartheta^\varphi \partial_\varphi Y_{(k)} -  v^\varphi q^\varphi g_{\varphi \varphi} Y_{(k)} \right] \dVS.
\end{equation}
We rewrite the last expression in coordinate form for the coefficients and use the notation $ \vartheta = \vartheta^\varphi $, $ q = q^\varphi $, $ v = -v^\varphi $ and $ g_{\varphi \varphi} = 1 $ to find
\begin{equation}
	\partial_t \rho_{(k)} = - \partial_i \int_{S^1}  \rho v^i Y_{(k)} \dS + \int_{S^1} \left[ \rho \vartheta \partial_\varphi Y_{(k)} +  v q Y_{(k)} \right] \dS.
\end{equation}
Also in the evolution of the curvature coefficients forms for $ k \geq 1 $ we only need to consider the different prefactor and the one-dimensional nature of the directional space to find
\begin{equation}
	\partial_t q_{(k)} dV = - \frac{1}{|k|} \dex \int_{S^1} \left[ q^\varphi \partial_\varphi Y_{(k)} \i_\V \dVS - \vartheta^\varphi \partial_\varphi Y_{(k)} \i_\R \dVS \right],
\end{equation}
such that we obtain the coordinate expression
\begin{equation}
	\partial_t q_{(k)} = - \frac{1}{|k|} \partial_i \int_{S^1} \left[ q v^i \partial_\varphi Y_{(k)}  \dS - \rho \vartheta l^i \partial_\varphi Y_{(k)} \dS \right].
\end{equation}
The evolution equation of the zeroth curvature coefficient misses the prefactor but otherwise largely conforms with the form of the other coefficient, i.e.,
\begin{equation}
	\partial_t q_{(0)} = -  \partial_i \int_{S^1} \left[ q v^i \dS - \rho \vartheta l^i \dS \right].
\end{equation}
\end{proof}
We note that $ q_0 $ is a conserved quantity, the integral of which is indicative of the total number of dislocations. This is a consequence of Hopf's Umflaufsatz (see e.g. \citet{pressley01}) which states that the integral over the scalar curvature of any closed planar curve without self intersections is $ 2 \mathrm{\pi} $. As $ q_{(0)} $ derives from averaging (i.e. adding) the curvature of the dislocations, the integral over it equals the total number of dislocations times $ 2 \mathrm{\pi} $. A similar argument directly employing the solenoidality of $ \R $ may be found in \citet{hochrainer13b}.

We again introduce a symbolic reformulation of Proposition \ref{Evolution Vector Fourier} which reads
\begin{align} \label{Eq: symbolic evolution coefficients Fourier}
	\partial_t \rho_\mm  &=  - \divergence \int_{S^1} Y_\mm \rho \v \dS + \int_{S^1} \left[ \rho \V \cdot \nabla  Y_\mm +  v q  Y_\mm \right] \dS \\
	\partial_t q_\mm &= - \divergence \int_{S^2} \left[ \R \cdot \Psibf_\mm \v - \V \cdot \Psibf_\mm \rho \l \right] \dS.
\end{align}

\subsection{Evolution of the irreducible tensors} \label{Sec: Evolution irreducible}
We now move on to the evolution equations of the tensor coefficients. The first set of equations are subject of
\begin{proposition} \label{Evolution irreducible} The evolution of the irreducible tensor coefficients of the SODT introduced in \emph{Definition \ref{Def: Irreducible tensor expansion}} is given by
\begin{align} \label{Eq: evolution irreducible} 
	\partial_t \hat{\rho}^{i_1 \cdots i_k}  &=  - \partial_i \int_{S^2}  \rho v^i \hat{l}^{i_1 \cdots i_k} \dS + \int_{S^2} \left[ \rho \vartheta^\mu \partial_\mu \hat{l}^{i_1 \cdots i_k} -  g_{\mu \nu} v^\mu q^\nu \hat{l}^{i_1 \cdots i_k} \right] \dS  \\
	\partial_t \hat{q}^{i_1 \cdots i_k} &= - \partial_i \int_{S^2} \left[  v^i g_{\mu \nu} q^\mu \hat{\Gamma}^{\nu i_1 \cdots i_k} - \rho l^i g_{\mu \nu} \vartheta^\mu \hat{\Gamma}^{\nu i_1 \cdots i_k}  \right] \dS  \\
	\partial_t \hat{q}^{* i_1 \cdots i_k} &= - \partial_i \int_{S^2} \left[ v^i g_{\mu \nu} q^\mu \hat{\Gamma}^{* \nu i_1 \cdots i_k} - \rho l^i \vartheta^\mu g_{\mu \nu} q^\mu \hat{\Gamma}^{* \nu i_1 \cdots i_k} \right] \dS - \sqrt{k\left(k+1\right)} \int_{S^2} \varepsilon_{\mu \nu} \vartheta^\mu q^\nu \hat{l}^{i_1 \cdots i_k} \dS. 
\end{align}
\end{proposition}
\begin{proof}
The base tensors $\hat{\mathbfit\Lambda}_{(k)} $, $\hat{\mathbfit\Gamma}_{(k)} $ and $\hat{\mathbfit\Gamma}^*_{(k)}$ are related by the same differential operations as $ \Y_\km $, $\Psibf_\km$ and $ \Psibf_\km $, respectively. Furthermore, the tensor coefficients are as well eigenfunctions of the Laplace-Beltrami operator to the according eigenvalue $-k(k-1)$. One may thus either proof the current proposition in analogy to Proposition \ref{Evolution VSH} or regard the current as a corollary to the latter.
\end{proof}
In symbolic notation the result of Proposition \ref{Evolution irreducible} takes the form
\begin{align}
	\partial_t \hat{\rhobf}_{(k)}  &=  - \divergence \int_{S^2} \v \otimes \R \cdot \hat{\mathbfit\Lambda}_{(k)} \dS + \int_{S^2} \left[ \rho \V \cdot \nabla \l^{\totimes k} -  \v \cdot \q \, \l^{\totimes k} \right] \dS  \label{Eq: evolution coefficients VSH symbol 1} \\
	\partial_t \hat{\q}_{(k)} &= - \divergence \int_{S^2} \left[ \v \otimes \R \cdot \hat{\mathbfit\Gamma}_{(k)}  -  \rho \l \otimes \V \cdot \hat{\mathbfit\Gamma}_{(k)}  \right] \dS \label{Eq: evolution coefficients VSH symbol 2} \\
	\partial_t \hat{\q}^*_{(k)} &= \left( - \divergence \int_{S^2} \left[ \v \otimes \R \cdot \hat{\mathbfit\Gamma}^*_{(k)} - \rho \l \otimes \V \cdot \hat{\mathbfit\Gamma}^*_{(k)} \right] \dS - \sqrt{k\left(k+1\right)} \int_{S^2} \thetabf \cdot \q^\perp \l^{\totimes k} \, \dS \right). \label{Eq: evolution coefficients VSH symbol 3}
\end{align}
The evolution of the (vector) Fourier coefficients in the case of glide only are given in
\begin{proposition} \label{Evolution irreducible (pure glide)} The evolution of the irreducible tensor coefficients in pure glide introduced in \emph{Definition \ref{Irreducible tensor expansion (pure glide)}} is given by
\begin{align} \label{Eq: evolution Fourier coefficients} 
	\partial_t \hat{\rho}^{i_1 \cdots i_k}  &=  - \partial_i \int_{S^1}  \rho v^i \hat{l}^{i_1 \cdots i_k} \dS + \int_{S^2} \left[ \rho \vartheta \partial_\varphi \hat{l}^{i_1 \cdots i_k} +  v q \hat{l}^{i_1 \cdots i_k} \right] \dS  \\
	\partial_t q_{(0)} &= -  \partial_i \int_{S^1} \left[ q v^i \dS - \rho \vartheta l^i \dS \right] \\
	\partial_t \hat{q}^{i_1 \cdots i_k} &= - \partial_i \int_{S^1} \left[ q v^i \hat{\Gamma}^{\varphi i_1 \cdots i_k} - \rho l^i \vartheta \hat{\Gamma}^{\varphi i_1 \cdots i_k}  \right] \dS, \qquad k \geq 1. 
\end{align}
\end{proposition}
\begin{proof}
The current proposition is a consequence of Proposition \ref{Evolution Vector Fourier} in the same way as Proposition \ref{Evolution irreducible} is a consequence of Proposition \ref{Evolution VSH}.
\end{proof}
In symbolic notation the formulae remain unaltered as compared to the general case, Eqs.\ \eqref{Eq: evolution coefficients VSH symbol 1} and \eqref{Eq: evolution coefficients VSH symbol 2} but lack the second curvature expansion:
\begin{align} \label{Eq: evolution Fourier coefficients symbol Appendix}
	\partial_t \hat{\rhobf}_{(k)}  &=  - \divergence \int_{S^2} \v \otimes \R \cdot \hat{\mathbfit\Lambda}_{(k)} \dS + \int_{S^2} \left[ \rho \V \cdot \nabla \l^{\totimes k} -  \v \cdot \q \, \l^{\totimes k} \right] \dS \\
		\partial_t \hat{\q}_{(0)} &= - \divergence \int_{S^1} \left[ q \v -  \rho \vartheta \l \right] \dS \\
	\partial_t \hat{\q}_{(k)} &= - \divergence \int_{S^1} \left[ \v \otimes \R \cdot \hat{\mathbfit\Gamma}_{(k)}  -  \rho \l \otimes \V \cdot \hat{\mathbfit\Gamma}_{(k)}  \right] \dS, \qquad k \geq 1
\end{align}
\subsection{Evolution of the reducible tensors} \label{Sec: Evolution reducible}
Also the evolution equations for the reducible tensors are easily derived from the evolution of $ \kappaII$. This is the content of 
\begin{proposition} The evolution of the reducible tensor coefficients introduced in \emph{Definition \ref{Def: Reducible tensor expansion}} is given by
\begin{align} \label{Eq: evolution reducible} 
	\partial_t \rho^{i_1 \cdots i_k}  &=  - \partial_i \int_{S^2}  \rho v^i l^{i_1} \cdots l^{i_k} \dS + \int_{S^2} \left[ \rho \vartheta^\mu \partial_\mu \left( l^{i_1} \cdots l^{i_k} \right) -  g_{\mu \nu} v^\mu q^\nu l^{i_1} \cdots l^{i_k} \right] \dS  \\
	\partial_t q^{i_1 \cdots i_k} &= - \partial_i \int_{S^2} \left[  v^i g_{\mu \nu} q^\mu \Gamma^{\nu i_1 \cdots i_k} - \rho l^i g_{\mu \nu} \vartheta^\mu \Gamma^{\nu i_1 \cdots i_k}  \right] \dS  \\
	\partial_t q^{* i_1 \cdots i_k} &= - \partial_i \int_{S^2} \left[ v^i g_{\mu \nu} q^\mu \Gamma^{* \nu i_1 \cdots i_k} - \rho l^i \vartheta^\mu g_{\mu \nu} q^\mu \Gamma^{* \nu i_1 \cdots i_k} \right] \dS - \nonumber \\ &{} \qquad \qquad \qquad \qquad \qquad \qquad \qquad \int_{S^2} \varepsilon_{\mu \nu} \vartheta^\mu q^\nu  \frac{1}{\sqrt{|g|}} \partial_\lambda \left[ \sqrt{|\g|} g^{\lambda \eta} \partial_\eta \left( l^{i_1} \cdots l^{i_k} \right) \right] \dS,
\end{align}
where $ |\g| = \det(\g) = \sin \theta$.
\end{proposition}
\begin{proof}
The base tensors $\mathbfit\Lambda_{(k)} $, $\mathbfit\Gamma_{(k)} $ and $\mathbfit\Gamma^*_{(k)}$ are again related by the same differential operations as $ \Y_\km $, $\Psibf_\km$ and $ \Psibf_\km $, respectively. However, the tensor coefficients are no eigenfunctions of the Laplace-Beltrami operator, which therefore occurs in the last integral. Aside from this, the result is again a consequence of Proposition \ref{Evolution VSH}.
\end{proof}
The symbolic version for the reducible tensors reads
\begin{align} \label{Eq: evolution coefficients VSH symbol Appendix}
	\partial_t \rhobf_{(k)}  &=  - \divergence \int_{S^2} \v \otimes \R \cdot \mathbfit\Lambda_{(k)} \dS + \int_{S^2} \left[ \rho \V \cdot \nabla \l^{\otimes k} -  \v \cdot \q \, \l^{\otimes k} \right] \dS \\
	\partial_t \q_{(k)} &= - \divergence \int_{S^2} \left[ \v \otimes \R \cdot \mathbfit\Gamma_{(k)}  -  \rho \l \otimes \V \cdot \mathbfit\Gamma_{(k)}  \right] \dS \\
	\partial_t \q^*_{(k)} &= \left( - \divergence \int_{S^2} \left[ \v \otimes \R \cdot \mathbfit\Gamma^*_{(k)} - \rho \l \otimes \V \cdot \mathbfit\Gamma^*_{(k)} \right] \dS - \int_{S^2} \thetabf \cdot \q^\perp \divergence \mathbfit\Gamma_{(k)} \, \dS \right).
\end{align}
The evolution equations for the reducible tensors in the case of glide only are subject of
\begin{proposition} \label{Prop: evolution reducible tensors (pure glide)} The evolution of the reducible tensor coefficients in the case of pure glide introduced in \emph{Definition \ref{Def: Reducible tensor expansion (pure glide)}} is given by
\begin{align}
	\partial_t \rho^{i_1 \cdots i_k}  &=  - \partial_i \int_{S^1}  \rho v^i l^{i_1} \cdots l^{i_k} \dS + \int_{S^1} \left[ \rho \vartheta \partial_\varphi \left( l^{i_1} \cdots l^{i_k} \right) + v q l^{i_1} \cdots l^{i_k} \right] \dS  \label{Eq: evolution reducible pg rho} \\
	\partial_t q_{(0)} &= -  \partial_i \int_{S^1} \left[ q v^i \dS - \rho \vartheta l^i \dS \right] \label{Eq: evolution reducible pg q0}\\
	\partial_t q^{i_1 \cdots i_k} &= - \partial_i \int_{S^1} \left[  q v^i \Gamma^{\varphi i_1 \cdots i_k} - \rho l^i \vartheta \Gamma^{\varphi i_1 \cdots i_k}  \right] \dS, \qquad k \geq 1. \label{Eq: evolution reducible pg q}
\end{align}
\end{proposition}
\begin{proof}
This result is a consequence of Proposition \ref{Evolution Vector Fourier}.
\end{proof}
In symbolic notation we rewrite the last Proposition as
\begin{align} \label{Eq: evolution reducible symbol Appendix}
	\partial_t \rhobf_{(k)}  &=  - \divergence \int_{S^1} \v \otimes \R \cdot \mathbfit\Lambda_{(k)} \dS + \int_{S^1} \left[ \rho \V \cdot \nabla \l^{\otimes k} +  v q \, \l^{\otimes k} \right] \dS \\
		\partial_t \hat{\q}_{(0)} &= - \divergence \int_{S^1} \left[ q \v -  \rho \vartheta \l \right] \dS \\
	\partial_t \q_{(k)} &= - \divergence \int_{S^1} \left[ \v \otimes \R \cdot \mathbfit\Gamma_{(k)}  -  \rho \l \otimes \V \cdot \mathbfit\Gamma_{(k)}  \right] \dS, \qquad k \geq 1.
\end{align}
\subsection{Evolution of the reducible differential forms} \label{Sec: Evolution Diff forms}
When regarding the evolution equations for the reducible first order alignment tensor (which is the same in the irreducible and the reducible definition) derived in the last Subsection,
\begin{equation} \label{Eq: evolution reducible alternative} 
	\partial_t \rho^{i}  =  - \partial_j \int_{S^2}  \rho v^j l^i \dS + \int_{S^2} \left[ \rho \vartheta^\mu \partial_\mu l^i -  g_{\mu \nu} v^\mu q^\nu l^i \right] \dS,
\end{equation}
it seems notable that this is not the well known evolution equation which dates at least back to \citet{mura63}, stating
\begin{equation}
	\partial_t \rho^{i} = - \varepsilon_{ijk} \partial_j \int_{S^2} \rho \varepsilon_{kmn} v^n l^m \dS.
\end{equation}	
One may indeed transform the last two equations into one another by using the solenoidality of $ \R $ and the definition of $ \thetabf $. Alternatively, we obtain a full set of evolution equations for the tensors $ \rhobf_{(k)} $ for $ k \geq 1 $ from the evolution of the tensors $ \mathbfit{K}_{(k)} $ which are provided in
\begin{proposition}
The evolution for the components of the tensor valued differential forms introduced in \emph{Definition \ref{Reducible differential form expansion}} is given by
\begin{multline}
	\partial_t \tensor{K}{_{ij} ^{i_1 \cdots i_k}} = - \left( \partial_i \int_{S^2} \rho \varepsilon_{jmk} v^k l^m l^{i_1} \cdots l^{i_k} \dS - \partial_j \int_{S^2} \rho \varepsilon_{imk} v^k l^m l^{i_1} \cdots l^{i_k} \dS \right) - \\ 
	 \int_{S^2} \left\{  \varepsilon_{ijk} v^k q^\mu \partial_\mu \left( l^{i_1} \cdots l^{i_k} \right) - \varepsilon_{ijk} \rho l^k \vartheta^\mu \partial_\mu \left( l^{i_1} \cdots l^{i_k} \right) \right\} \dS.
\end{multline}
\end{proposition}
\begin{proof}
The derivation is straight forward in symbolic notation: 
\begin{align}
	\partial_t \mathbfit{K}_{(k)} =& - \int_{S^2} \dex \i_\Vel \kappaII \otimes \l^{\otimes k}  \nonumber \\
		=& - \int_{S^2} \left[ \dex \left( \i_\Vel \kappaII \otimes \l^{\otimes k} \right) - (-1)^3\i_\Vel \kappaII \wedge \nabla \l^{\otimes k} \right] \nonumber \\
		=& - \dex \int_{S^2} \left( \i_\Vel \kappaII \otimes \l^{\otimes k} \right)  - \int_{S^2} \left\{ \i_\Vel \left[ \kappaII \wedge \nabla \l^{\otimes k} \right] - (-1)^4 \kappaII \wedge \i_\Vel \nabla  \l^{\otimes k} \right\} \nonumber \\
		=& - \dex  \int_{S^2} \left( \i_\Vel \i_\R \dVS \otimes \l^{\otimes k} \right)  - \int_{S^2} \left\{ \i_\Vel \dVS \otimes q^\mu \partial_\mu \l^{\otimes k}  - \i_\R \dVS \otimes \vartheta^\mu \partial_\mu  \l^{\otimes k} \right\}
\end{align}
With regard to the last step we note that
\begin{equation}
	\kappaII \wedge \nabla \l^{\otimes k} = \left( -1 \right)^{(4\cdot1)} \nabla \l^{\otimes k} \wedge \kappaII = \R \cdot \nabla \l^{\otimes k} \dVS = q^\mu \partial_\mu \l^{\otimes k} \dVS.
\end{equation}
The translation into index notation follows from the calculus of differential forms.
\end{proof}
By employing the relation $ \rho^{i_1 \cdots i_k} = 0.5 \varepsilon^{i_1 ij} \tensor{K}{_{ij} ^{i_2 \cdots i_k}} $ we obtain the alternative evolution equation for the reducible tensor coefficients
\begin{align}
	\partial_t \rho^{i_1 \cdots i_k}		=& - \varepsilon^{i_1ij} \partial_i \int_{S^2} \rho  \varepsilon_{jmk} v^k l^m l^{i_2} \cdots l^{i_k} \dS  - \int_{S^2} \left(  v^{i_1} q^\mu \partial_\mu \left( l^{i_2} \cdots l^{i_k} \right)  -   \rho l^{i_1} \vartheta^\mu \partial_\mu \left( l^{i_2} \cdots l^{i_k} \right)  \right) \dS,
\end{align}
which holds for $ k \geq 1 $. From this we read the symbolic version
\begin{align}
	\partial_t \rhobf_{(k)}		=& - \nabla \times \int_{S^2} \rho  \l \times \v \otimes \l^{\otimes (k-1)} \dS  - \int_{S^2} \left( \v \otimes  \R \cdot \nabla \l^{\otimes (k-1)} - \rho \l \otimes \V \cdot \nabla \l^{\otimes (k-1)}  \right) \dS.
\end{align}
We note that in the special case $ k = 1 $ only the curl part is non-zero and the equation recovers Eq.\ \eqref{Eq: evolution kappa intro}. The case $ k = 0 $ cannot be covered by this derivation.

The evolution of the coefficients of the tensor valued differential forms in the case of glide only is subject of
\begin{proposition}
In the case of glide only the evolution for the components of the tensors introduced in \emph{Definition \ref{Reducible differential form expansion}} is given by
\begin{multline}
	\partial_t \tensor{K}{_{ij} ^{i_1 \cdots i_k}}		= - \left( \partial_i \int_{S^1} \rho \varepsilon_{jmk} v^k l^m l^{i_1} \cdots l^{i_k} \dS - \partial_j \int_{S^1} \rho \varepsilon_{imk} v^k l^m l^{i_1} \cdots l^{i_k} \dS \right) - \\ 
 \int_{S^1} \left\{  \varepsilon_{ijk} v^k q \partial_\varphi \left( l^{i_1} \cdots l^{i_k} \right) - \varepsilon_{ijk} \rho l^k \vartheta \partial_\varphi  l^{i_1} \cdots l^{i_k} \right\} \dS,
\end{multline}
\end{proposition}
\begin{proof}
A few signs get flipped in comparison to the general case: 
\begin{align}
	\partial_t \mathbfit{K}_{(k)} =& - \int_{S^2} \dex \i_\Vel \kappaII \otimes \l^{\otimes k} \nonumber \\
		=& - \int_{S^2} \left[ \dex \left( \i_\Vel \kappaII \otimes \l^{\otimes k} \right) - (-1)^2\i_\Vel \kappaII \wedge \nabla \l^{\otimes k} \right] \nonumber \\
		=& - \dex \int_{S^2} \left( \i_\Vel \kappaII \otimes \l^{\otimes k} \right)  + \int_{S^2} \left\{ \i_\Vel \left[ \kappaII \wedge \nabla \l^{\otimes k} \right] - (-1)^3 \kappaII \wedge \i_\Vel \nabla  \l^{\otimes k} \right\} \nonumber \\
		=& - \dex  \int_{S^1} \left( \i_\Vel \i_\R \dVS \otimes \l^{\otimes k} \right)  - \int_{S^1} \left\{ \i_\Vel \dVS \otimes q \partial_\varphi \l^{\otimes k} -  \i_\R \dVS \otimes \vartheta \partial_\varphi  \l^{\otimes k}  \right\} 
\end{align}
With regard to the last step we note that
\begin{equation}
	\kappaII \wedge \nabla \l^{\otimes k} = \left( -1 \right)^{(3\cdot1)} \nabla \l^{\otimes k} \wedge \kappaII = - \R \cdot \nabla \l^{\otimes k} \dVS = - q \partial_\varphi \l^{\otimes k} \dVS,
\end{equation}
which finishes the proof.
\end{proof}
We again translate this by employing $ \rho^{i_1 \cdots i_k} = 0.5 \varepsilon^{i_1 ij} \tensor{K}{_{ij} ^{i_2 \cdots i_k}} $ into the alternative representation for the evolution of the coefficients $ \rho^{i_1 \cdots i_k} $ for $ k \geq 1 $ as
\begin{align} \label{Eq: evolution rho from K}
	\partial_t \rho^{i_1 \cdots i_k}		=& - \varepsilon^{i_1ij} \partial_i \int_{S^1} \rho  \varepsilon_{jmk} v^k l^m l^{i_2} \cdots l^{i_k} \dS  - \int_{S^1} \left(  v^{i_1} q \partial_\varphi \left( l^{i_2} \cdots l^{i_k} \right)  -   \rho l^{i_1} \vartheta \partial_\varphi \left( l^{i_2} \cdots l^{i_k} \right)  \right) \dS.
\end{align}
For the symbolic notation we note that $ \l \times \v = - v \n $ with the glide plane normal $ \n $ and we obtain
\begin{align}
	\partial_t \rhobf_{(k)}		=& \nabla \times \int_{S^1} \rho  v \n \otimes \l^{\otimes (k-1)} \dS  - \int_{S^1} \left( \v \otimes  \R \cdot \nabla \l^{\otimes (k-1)} - \rho \l \otimes \V \cdot \nabla \l^{\otimes (k-1)}  \right) \dS, \quad k \geq 1. 
\end{align} 
The evolution equations presented in this section actually define an infinite hierarchy of equations: that is, the evolution of the $k$-th order tensors typically involves information from the higher order tensors. This means that for practical purposes we did not yet improve the situation of evolving the dislocation state as compared to solving the higher dimensional theory. In order to arrive at a manageable theory one seeks to work with only a few low order tensors. This requires closure assumptions in the sense that the unknown higher order tensors appearing in the evolution equation of the regarded tensors need to be approximated using information of the lower order tensors. However, before dealing with the closure assumptions we need to make assumptions or rather approximations for the orientation dependent velocity $ \v $ and their lift $ \V = \v + \mathbfit\vartheta $.
\section{Expansion of the velocity} \label{Sec: Velocity expansion}
The lifted velocity $ \V $ is an orientation dependent vector field on the configuration space $ U = M \times S^2 $. It is composed of a spatial velocity $ \v $ and the rotational velocity vector $ \mathbfit\vartheta $. Both parts are vectors orthogonal to the canonical line direction $ \l $. Therefore, both can be expanded in terms of the vector spherical harmonics $ \Psibf_\km $ and $ \Phibf_\km $ or the basis tensors $ \hat{\mathbfit\Gamma}_{(k)} $ and $ \hat{\mathbfit\Gamma}^*_{(k)} $. Note, however, that for the expansion of $ \v $ we need to define horizontal versions of theses fields, which we introduce as
\begin{align}
	\Psibf_\km^\mathrm{h} &= \frac{1}{\sqrt{k\left(k+1\right)} } g^{\mu \nu} \partial_\nu Y_\km \e_{\mu} \\
	\Phibf_\km^\mathrm{h} &= \frac{1}{\sqrt{k\left(k+1\right)} } \varepsilon^{\mu \nu} \partial_\nu Y_\km \e_{\mu}
\end{align}
in case of the vector spherical harmonics and as
\begin{align}
	\hat{\mathbfit\Gamma}_{(k)}^\mathrm{h} &= \frac{1}{\sqrt{k(k+1)}} g^{\mu \nu} \partial_\nu \lhat^{i_1 \cdots i_k} \e_{\mu} \otimes \bpartial_{i_1} \otimes \cdots \otimes \bpartial_{i_k} \\ 
	\hat{\mathbfit\Gamma}^{\mathrm{h}*}_{(k)} &= \frac{1}{\sqrt{k(k+1)}} \varepsilon^{\mu \nu} \partial_\nu \lhat^{i_1 \cdots i_k} \e_{\mu} \otimes \bpartial_{i_1} \otimes \cdots \otimes \bpartial_{i_k}.
\end{align} 
for the base tensors. We can then give the expansion of $ \V $ which reads
\begin{equation} \label{Eq: expansion V}
	\V =\sum_{k=0}^\infty\sum_{m=-k}^k\left(v_\km \Psibf_\km^\mathrm{h}+v^*_\km \Phibf_\km^\mathrm{h}+\vartheta_\km \Psibf_\km+\vartheta^*_\km \Phibf_\km\right),
\end{equation}
where the coefficients derive as
\begin{align}
	v_\km &= \int_{S^2} \V \cdot \Psibf_\km^\mathrm{h} \dS, \qquad	v^*_\km = \int_{S^2} \V \cdot \Phibf_\km^\mathrm{h} \dS \\
	\vartheta_\km &= \int_{S^2} \R \cdot \Psibf_\km \dS, \qquad \vartheta^*_\km = \int_{S^2} \R \cdot \Phibf_\km \dS.
\end{align}
As a tensor expansion we find
\begin{multline}
	\V(\l) = \frac{1}{4\mathrm{\pi}} \left[ \sum_{k=1}^\infty  \frac{\left(2k+1\right)!!}{k!} \left( \hat{v}_{i_1 \cdots  i_k} \hat{\Gamma}^{\mathrm{h} \, i_1 \cdots i_k \mu } \e_\mu  + \hat{v}^{*}_{i_1 \cdots  i_k} \hat{\Gamma}^{\mathrm{h}* \, i_1 \cdots i_k \mu } \e_\mu + \right. \right. \\  \left. \left. \hat{\vartheta}_{i_1 \cdots  i_k} \hat{\Gamma}^{i_1 \cdots i_k \mu } \bpartial_\mu  + \hat{\vartheta}^{*}_{i_1 \cdots  i_k} \hat{\Gamma}^{*i_1 \cdots i_k \mu } \bpartial_\mu \right) \right].
\end{multline}
with the tensor coefficients obtained from
\begin{align}
	\hat{\v}_{(k)} &=  \int_{S^2} \R \cdot \hat{\mathbfit\Gamma}_{(k)}^\mathrm{h} \dS, \qquad \hat{\v}^{*}_{(k)} = \int_{S^2} \R \cdot \hat{\mathbfit\Gamma}^{\mathrm{h}*}_{(k)} \dS \\
	\hat{\mathbfit\vartheta}_{(k)} &=  \int_{S^2} \R \cdot \hat{\mathbfit\Gamma}_{(k)} \dS, \qquad \hat{\mathbfit\vartheta}^{*}_{(k)} = \int_{S^2} \R \cdot \hat{\mathbfit\Gamma}^*_{(k)} \dS. 
\end{align} 

So far we derived independent expansions of $ \v $ and $ \mathbfit\vartheta $. However the rotational velocity $ \mathbfit\vartheta $ derives from $ \v $ as 
\begin{equation}
   \thetabf_\mathrm{h} = \left( \nabla_\L \v \right)^\perp= \nabla_\L \v - \g \left(  \nabla_\L \v, \l \right) \l,
\end{equation}
such that we may define an alternative expansion of $  \mathbfit\vartheta $ from the expansion of $ \v $ as
\begin{equation} \label{Eq: expansion theta alternative}
   \thetabf_\mathrm{h} = \sum_{k=1}^\infty \thetabf^\mathrm{h}_{(k)} = \sum_{k=1}^\infty \frac{\left(2k+1\right)!!}{k!}  \left[ \nabla_\L \left( \hat{v}_{i_1 \cdots  i_k} \hat{\Gamma}^{\mathrm{h} \, i_1 \cdots i_k \mu } \e_\mu  + \hat{v}^{*}_{i_1 \cdots  i_k} \hat{\Gamma}^{\mathrm{h}* \, i_1 \cdots i_k \mu } \e_\mu \right)^\perp \right]. 
\end{equation}
The exact relation between this series expansion and $ \hat\thetabf_{(k)} $ will have to be established in future work. But it seems natural to us that when working with a truncated expansion of the velocity, the latter expansion ensures consistency when considering the same number of elements in the expansion of $ \thetabf $. In the following we will only use expansion Eq.\ \eqref{Eq: expansion theta alternative} for $\thetabf$.

We conclude this section by noting that the reducible expansion of the velocity in the general case and the treatment of the case of only conservative dislocation motion are obvious from the definitions introduced so far. We will therefore not go in any detail in this regard.  

\section{Closure approximations} \label{Sec: Closure approximations}
As noted at the end of the last but one Section, a useful plasticity theory can only be achieved if the tensor expansion is truncated at low order. Such a truncation will always involve two distinct decisions: the first decision is to which level of detail the dislocation velocity has to be modeled; the second decision is which is the highest order dislocation density tensor to be considered. The angular dependent dislocation velocity will in general have a complex and possibly non-linear dependence on the current dislocation state. To model this relation is a very challenging part of continuum dislocation modeling which we will only briefly review in the Discussion of this paper. In any case, the decision for truncating the velocity expansion involves the crystal structure and possibly information on the microstructure (part of which is the dislocation density description). The crystal structure (and temperature) will for example tell us whether it makes sense to consider a nearly isotropic velocity (in most fcc materials) or a markedly different velocity for dislocations of different types (screw and edges) as observed, e.g., in bcc materials. In the first case it may be sensible to only consider the zeroth order terms of the velocity (at least in the pure glide case) while in the bcc case one needs to include the second order velocity tensor. We note that the velocity will typically be symmetric (dislocations of opposite direction move with the same velocity in opposite directions) such that only even order velocity tensors need to be considered. In the following we will only discuss closure assumptions for the pure glide case. Meaningful assumptions for non-conservative dislocation motion remain a topic of future research.
\subsection{Preliminaries in planar dislocation motion} \label{Sec: Preliminaries}
We begin with a few notations specific to the case of planar dislocation motion in a plane perpendicular to the unit normal $ \n $. Because the dislocation velocity $ \v $ is always perpendicular to the dislocation line direction $ \l $ and lies within the glide plane we have
\begin{equation} 
	\v = - v \n \times \l = - v \e_\varphi = - v \partial_\varphi \l.
\end{equation}
In order to shorten the index notation we introduce the tensor $ \varepsilon_{ij} = \varepsilon_{ikj} n^k $ such that
\begin{equation}
	v^i = - v \tensor{\varepsilon}{^{i}_{kj}} n^k l^j = - v \tensor{\varepsilon}{^{i}_{j}} l^j.
\end{equation}
Furthermore, we note that we have
\begin{equation}
	\varepsilon_{jmk} v^m l^k = v n_j,
\end{equation}
such that we rewrite the evolution equations of Proposition \ref{Prop: evolution reducible tensors (pure glide)} (using variant Eq.\ \eqref{Eq: evolution rho from K} for the expansion of $ \rho$) for the reducible tensor coefficients in pure glide  as
\begin{align} \label{Eq: evolution reducible pure glide} 
	\partial_t \rho_{(0)}  &=  - \partial_i \int_{S^2}  - v \rho  \tensor{\varepsilon}{^{i}_{j}} l^j \dS + \int_{S^1} v q \dS \\
	\partial_t \rho^{i_1 \cdots i_k}		&= \varepsilon^{i_1ij} \partial_i \int_{S^1} v \rho n_j l^{i_2} \cdots l^{i_k} \dS  - \int_{S^1} \left(  - vq \tensor{\varepsilon}{^{i_1}_{j}} l^j \partial_\varphi \left( l^{i_2} \cdots l^{i_k} \right)  -   \rho \vartheta l^{i_1} \partial_\varphi \left( l^{i_2} \cdots l^{i_k} \right)  \right) \dS, \qquad k \geq 1. \label{Eq: evolution rho i1...ik}\\
	\partial_t q_{(0)} &= -  \partial_i \int_{S^1} \left( - v q \tensor{\varepsilon}{^{i}_{j}} l^j \dS - \rho \vartheta l^i \right) \dS \\
	\partial_t q^{i_1 \cdots i_k} &= - \frac{1}{k} \partial_i \int_{S^1} \left( - v q \tensor{\varepsilon}{^{i}_{j}} l^j \partial_\varphi \left( l^{i_1} \cdots l^{i_k} \right) - \rho \vartheta l^i \partial_\varphi \left( l^{i_1} \cdots l^{i_k} \right)  \right) \dS, \qquad k \geq 1. 
\end{align}
By derivation the right hand side of the evolution equation (\ref{Eq: evolution rho i1...ik}) is symmetric in all indices although this is not true for the single terms. But this means, that the equation remains valid when we perform a symmetrizing operation on the right hand side, which applies to each term. It remains valid as well, if we insert in the symmetrizing operation a non-symmetric expression with the same symmetric part. This is what we use when we now substitute
\begin{equation}
	\partial_\varphi \left( l^{i_1} \cdots l^{i_k} \right) = k \left( \partial_\varphi l^{i_1} l^{i_2} \cdots l^{i_k} \right)_{\mathrm{sym}} = k \left( \tensor{\varepsilon}{^{i_1}_{j}} l^j l^{i_2} \cdots l^{i_k} \right)_{\mathrm{sym}},
\end{equation}
because $ \tensor{\varepsilon}{^{i}_{j}} l^j = \partial_\varphi l^i $ this provides us with the following form of the evolution equations (the zeroth order terms remain unaltered)
\begin{align}
	\partial_t \rho^{i_1 \cdots i_k}		&= - \left[ \varepsilon^{i_1i} \partial_i \int_{S^1} v \rho l^{i_2} \cdots l^{i_k} \dS  - (k-1) \int_{S^1} \left(  - vq \tensor{\varepsilon}{^{i_1}_{j}} l^j \tensor{\varepsilon}{^{i_2}_{m}} l^m l^{i_3} \cdots l^{i_k}  - \right. \right. \label{Eq: evolution reducible rotation} \nonumber \\ &{} \left. \left. \qquad \qquad \qquad \qquad \qquad \qquad \qquad \qquad \qquad \qquad \qquad \rho \vartheta l^{i_1} \tensor{\varepsilon}{^{i_2}_{m}} l^m l^{i_3} \cdots l^{i_k}  \right) \dS \right]_{\mathrm{sym}} \\
	\partial_t q^{i_1 \cdots i_k} &= \left[ - \partial_i \int_{S^1} \left(  - v q \tensor{\varepsilon}{^{i}_{j}} l^j \tensor{\varepsilon}{^{i_1}_{m}} l^m l^{i_2} \cdots l^{i_k} - \rho \vartheta l^i \tensor{\varepsilon}{^{i_1}_{m}} l^m l^{i_2} \cdots l^{i_k}  \right) \dS \right]_{\mathrm{sym}}
\end{align}
Note that the normal $ \n $ is of course independent of orientation and can be pulled out of the integrals. In doing so we additionally introduced the tensor $ \varepsilon^{ij} = \varepsilon^{ikj} n_k $. With the latter form of the evolution equations we will be easily able to recognize the higher order tensors appearing in the hierarchy of equations under simplifying assumptions on the dislocation velocity.

For completeness we also provide the evolution equations for the density tensors as obtained from Proposition \ref{Prop: evolution reducible tensors (pure glide)}. We will discuss why these equations are less well suited for closing the hierarchy at low order. Given the current preparation for glide only we find from Proposition \ref{Prop: evolution reducible tensors (pure glide)}
\begin{equation} \label{Eq: evolution reducible divergence} 
	\partial_t \rho^{i_1 \cdots i_k}  =  -  \left[ \partial_i \int_{S^1}  v \rho \tensor{\varepsilon}{^{i}_{j}} l^j l^{i_1} \cdots l^{i_k} \dS - \int_{S^1} \left( k \vartheta \rho \tensor{\varepsilon}{^{i_1}_{j}} l^j l^{i_2} \cdots l^{i_k} + v q l^{i_1} \cdots l^{i_k} \right) \dS \right]_{\mathrm{sym}}
\end{equation}
We note that in the divergence term of the evolution of the $k$-th order tensor the line direction $ \l $ appears $ k+1 $ times while in the corresponding rotation term of Eq.\ \eqref{Eq: evolution reducible rotation} $ \l $ only appears $ k-1 $ times. We may therefore anticipate that all terms in Eq.\ \eqref{Eq: evolution reducible divergence} will require information on higher order tensors and therefore are subject to closure approximations regardless of the assumptions made for the dislocation velocity.
\subsection{Isotropic dislocation velocity} \label{Isotropic}
We begin with the most restrictive assumption on the velocity, namely that the magnitude of the dislocation velocity does not depend on the segment orientation. This is for example employed in DDD simulations of fcc-materials (cf.\ \citet{weygand_etal02}) and underlies most other dislocation density based models as well. We note that this assumption only makes sense in the case of planar dislocation motion. Because now $ v $ does not depend on the orientation $ \varphi$  we can pull $ v $ out of all the integrals in the evolution equations derived in the last subsection. Furthermore, we note that the rotational velocity takes on the simple shape $ \vartheta = - l^i \partial_i v $ and also the partial derivatives $ \partial_i v $ can be pulled out of the integrals. We thus find
\begin{align} \label{Eq: evolution reducible pure glide iso v} 
	\partial_t \rho_{(0)}  &=  \partial_i \left( v \int_{S^2}  \rho  \tensor{\varepsilon}{^{i}_{j}} l^j \dS \right) + v \int_{S^1} q \dS  \\
	\partial_t \rho^{i_1 \cdots i_k}		&= \left[ - \varepsilon^{i_1 i} \partial_i \left( v \int_{S^1} \rho l^{i_2} \cdots l^{i_k} \dS \right) + (k-1) v \int_{S^1} q \tensor{\varepsilon}{^{i_1}_{j}} l^j \tensor{\varepsilon}{^{i_2}_{m}} l^m l^{i_3} \cdots l^{i_k} \dS  - \right.  \nonumber \\ 
	&{} \qquad \qquad \qquad \qquad \qquad \qquad \qquad \qquad \qquad \left. (k-1) \int_{S^1} \rho l^i l^{i_1} \tensor{\varepsilon}{^{i_2}_{m}} l^m l^{i_3} \cdots l^{i_k}  \dS \; \partial_i v \right]_{\mathrm{sym}}  \\
	\partial_t q_{(0)} &=  \partial_i \left( v \int_{S^1} q \tensor{\varepsilon}{^{i}_{j}} l^j \dS - \int_{S^1} \rho l^j l^i \dS \partial_j v \right)  \\
	\partial_t q^{i_1 \cdots i_k} &= \left[ \partial_i \left( v \int_{S^1} q \tensor{\varepsilon}{^{i}_{j}} l^j \tensor{\varepsilon}{^{i_1}_{m}} l^m l^{i_2} \cdots l^{i_k} \dS - \int_{S^1} \rho l^j l^i \tensor{\varepsilon}{^{i_1}_{m}} l^m l^{i_2} \cdots l^{i_k} \dS \; \partial_j v  \right) \right]_{\mathrm{sym}} .
\end{align}
In this form we can recognize the tensors appearing in the integrals and arrive at
\begin{align}  
	\partial_t \rho_{(0)}  &=  \partial_i \left( v \tensor{\varepsilon}{^{i}_{j}} \rho^j \right) + v q_{(0)} \label{Eq: evolution reducible pure glide iso v tensors 1} \\
	\partial_t \rho^{i_1 \cdots i_k}		&= \left[ - \varepsilon^{i_1 i} \partial_i \left( v \rho^{i_2 \cdots i_k} \right)  + (k-1) v Q^{i_1 \cdots i_k}  -  (k-1) \tensor{\varepsilon}{^{i_2}_{m}} \rho^{i_1 m i_3 \cdots i_k i} \;  \partial_i v \right]_{\mathrm{sym}} \label{Eq: evolution reducible pure glide iso v tensors 2} \\
	\partial_t q_{(0)} &=  \partial_i \left( v  q^i - \rho^{ji} \; \partial_j v \right) \label{Eq: evolution reducible pure glide iso v tensors 3} \\
	\partial_t q^{i_1 \cdots i_k} &= \left[ \partial_i \left( v Q^{i_1 \cdots i_k} - \tensor{\varepsilon}{^{i_1}_{m}} \rho^{i m i_2 \cdots i_k j} \; \partial_j v \right) \right]_{\mathrm{sym}} \label{Eq: evolution reducible pure glide iso v tensors 4}
\end{align}
Notably, the tensors $ Q^{i_1 \cdots i_k} $ containing the curvature information on the right hand side are not the curvature tensors $ q^{i_1 \cdots i_k} $ (except for $k=1$) but are defined as
\begin{align} 
	Q^{i_1} &= \int_{S^1} \left(  q \tensor{\varepsilon}{^{i_1}_{j}} l^j \right) \dS = q^{i_1} \label{Eq: Alt Q tensors 1} \\
	Q^{i_1 i_2} &= \int_{S^1} \left(  q \tensor{\varepsilon}{^{i_1}_{j}} l^j \tensor{\varepsilon}{^{i_2}_{m}} l^m \right) \dS \label{Eq: Alt Q tensors 2} \\
	Q^{i_1 \cdots i_k} &= \int_{S^1} \left(  q \tensor{\varepsilon}{^{i_1}_{j}} l^j \tensor{\varepsilon}{^{i_2}_{m}} l^m l^{i_3} \cdots l^{i_k} \right) \dS, \quad k \geq 3. \label{Eq: Alt Q tensors 3}
\end{align}
These tensors will usually be subject to closure assumptions or will have to be evolved separately. In the following we will introduce appropriate closure assumptions. Evolution equations for the tensors $ \Q $ were not yet derived.

We conclude this section with a short summary of the obtained hierarchy of evolution Eqs.\ \eqref{Eq: evolution reducible pure glide iso v tensors 1} -- \eqref{Eq: evolution reducible pure glide iso v tensors 4}. The evolution of the total dislocation density needs information on the first order density tensor and the total scalar curvature. The evolution of a tensor of order $k$ contains a flux term based on the tensor of order $ k-1 $, a curvature tensor of the same order $ k$ and a contraction of the next higher order tensor with the gradient of the velocity. The term with the curvature tensor accounts for line length increase while the one with the gradient of the velocity takes care of dislocation reorientation. The evolution of the curvature tensor of order $ k$ is obtained from the divergence of the alignment tensor of next higher order. It therefore contains a curvature tensor of the next higher order ($k+1$) and the alignment tensor of order $(k+2)$ contracted with the gradient of the velocity. A manageable continuum dislocation dynamics theory requires working with a few low order tensors. To achieve this one needs to express some unknown next higher order tensors with known ones within the evolution equations. 
\subsubsection{Lowest order closure} \label{Sec: Lowest order closure}

The lowest order theory still able to take care of dislocation fluxes and line length increase will contain $ \rho_{(0)} $, $ \kappavec = \rhobf_{(1)} $ and $ q_{(0)} $. The according evolution equations read
\begin{align} 
	\partial_t \rho_{(0)}  &=  \partial_i \left( v \tensor{\varepsilon}{^{i}_{j}} \kappa^j \right) + v q_{(0)} \label{Eq: evol rho iso v} \\
	\partial_t \kappa^{i}	&= -\varepsilon^{i j} \partial_j \left( v \rho \right)  \label{Eq: evol kappa iso v} \\
	\partial_t q_{(0)} &=  \partial_i \left( v  q^i - \rho^{ji} \partial_j v \right). \label{Eq: evol q iso v} 
\end{align}
Closure assumptions are therefore needed for $ \q_{(1)} $ and $ \rhobf_{(2)} $.

Regarding the closure assumptions we first note that we know the spherical part of $ \rhobf_{(2)} $ because the trace of it equals $ \rho_{(0)} $. The traceless part $ \hat\rhobf_{(2)} $, however, is not known and will be constructed from appropriate assumptions. The second oder dislocation density tensor $ \rhobf_{(2)} $ is a symmetric tensor and accordingly has two orthogonal eigendirections within the glide plane, while the third eigendirection is the glide plane normal $\n$ with eigenvalue $0$. In the current approach the net dislocation vector $ \kappavec $ is the only available directional information. We therefor assume the majority of the dislocations to be aligned with the net dislocation density vector. In other words, we assume the eigendirection of $ \rhobf_{(2)} $ to be aligned with $\kappavec$ and the orthogonal direction $ \kappavec^\perp = - \n \times \kappavec $. Note that in the case of only geometrically necessary dislocations (when all dislocation are aligned in one direction) the higher dimensional density is a Dirac delta distribution on the orientation space such that $\rhobf_{(2)} = (\kappavec \otimes \kappavec)/|\kappavec| $. In that case the (in-plane) traceless part is given by
\begin{equation} 
	\hat\rho^{ij} = \frac{1}{\sqrt{ \kappa_m \kappa^m }} \kappa^i \kappa^j - \frac{1}{2} \left(\delta^{ij} - n^i n^j \right) \sqrt{ \kappa_m \kappa^m},
\end{equation}
which may be rewritten as 
\begin{equation} 
	\hat\rho^{ij} = \frac{1}{2\sqrt{\kappa_m \kappa^m }} \left( \kappa^i \kappa^j - \tensor{\varepsilon}{^{i}_{k}} \kappa^k \tensor{\varepsilon}{^{j}_{l}} \kappa^l \right).
\end{equation}
In symbolic notation the last equation reads
\begin{equation} 
	\hat\rhobf_{(2)} = \frac{1}{2|\kappavec|} \left( \kappavec \otimes \kappavec - \kappavec^\perp \otimes \kappavec^\perp \right).
\end{equation}
As the special case of only GND needs to be included in any higher order CDD we take this equation as closure assumption for the traceless part. We summarize that we approximate the \emph{non-traceless} tensor of second order as
\begin{equation} \label{Eq: closure rho2}
	\rhobf_{(2)} := \frac{1}{2|\kappavec|^2} \left[ \left( \rho_{0} + |\kappavec| \right) \kappavec \otimes \kappavec + \left( \rho_{0} - |\kappavec| \right) \kappavec^\perp \otimes \kappavec^\perp \right] .
\end{equation}
We now turn to the closure assumption for the average curvature tensor $ \q_{(1)} $. Again, the only available directional information is the dislocation density vector $ \kappavec$. But the curvature vector is always perpendicular to the dislocation line direction, such that we assume the curvature vector orthogonal to the dislocation density vector. Furthermore, the length of the curvature vector should be given by the length of $ \kappavec$ times the average curvature of the dislocations $ \bar{k} $. Note that on the higher dimensional configuration space we interpret the quotient of $ q $ and $ \rho $ as the orientation dependent dislocation curvature $ k = q / \rho $ \citep{hochrainer_zg07}. Accordingly we define
\begin{equation} 
	\bar{k} := \frac{\int_{S^1}q\dS}{\int_{S^1} \rho \dS } = \frac{q_{(0)}}{\rho_{(0)}}.
\end{equation}
The curvature vector is consequently estimated as
\begin{equation} 
	\q_{(1)} := - \bar{k} \kappavec^\perp.
\end{equation}
In coordinate notation we obtain
\begin{equation} 
	q^i := \frac{q_{(0)}}{\rho_{(0)}} \tensor{\varepsilon}{^{i}_{j}} \kappa^j.
\end{equation}
Again, this assumption is exact in the case of only geometrically necessary dislocations.

Finally, we introduce the two assumptions into the evolution equation of $q_{(0)}$ (\ref{Eq: evol q iso v}) and obtain
\begin{equation} \label{Eq: evol q simplified}
	\partial_t q_{(0)} =  -\divergence \left( v  \bar{k} \kappavec^\perp + \frac{1}{2|\kappavec|^2} \left[ \left( \rho_{0} + |\kappavec| \right) \kappavec \otimes \kappavec + \left( \rho_{0} - |\kappavec| \right) \kappavec^\perp \otimes \kappavec^\perp \right] \cdot \nabla v \right), 
\end{equation}
which we additionally provide in coordinate notation,
\begin{equation} \label{Eq: evol q simplified coord} 
	\partial_t q_{(0)} =  - \partial_i \left( - v  \frac{q_{(0)}}{\rho_{(0)}} \tensor{\varepsilon}{^{i}_{j}} \kappa^j + \frac{1}{ 2\kappa_m \kappa^m } \left[ \left( \rho_{0} + \sqrt{ \kappa_m \kappa^m } \right) \kappa^i \kappa^j + \left( \rho_{0} - \sqrt{ \kappa_m \kappa^m } \right) \tensor{\varepsilon}{^{i}_{k}} \kappa^k \tensor{\varepsilon}{^{j}_{l}} \kappa^l \right] \partial_j v \right).
\end{equation}
With this closure we arrived at a system of equations for the evolution of planar systems of dislocations with isotropic mobility with only three internal variables, that are $ \rho_{(0)} $, $ \kappavec = \rhobf_{(1)} $ and $ q_{(0)} $. The evolution equations care for fluxes of all dislocations and produce dislocation line length (\ref{Eq: evol rho iso v}), are consistent with the classical continuum theory of dislocations (\ref{Eq: evol kappa iso v}) but keep the total number of dislocations fixed (\ref{Eq: evol q iso v}), as $ q_{(0)} $ is a conserved quantity. Because of the conservative form this equation lends itself for implementation with conservative numerical schemes, as for example, finite volume or discontinuous Galerkin methods. Numerical calculations with the current formulation have been presented by \citet{ebrahimi_mh14} and \citet{monavari_zs14}.

Before we move on to the closure at next higher order, we first recover an earlier version of this theory as presented in \citet{hochrainer_etal14}, which was derived in a somewhat more elementary manner. We start from (\ref{Eq: evol q iso v}) and use the fact that $ \q_{(1)} = \divergence \rhobf_{(2)} $, such that employing the product rule we find
\begin{align} \label{Eq: transformation q v iso} 
	\partial_t q_{(0)} &=  \partial_i \left( v  q^i - \rho^{ji} \partial_j v \right)  \nonumber  \\
										 &=  q^i \partial_i v + v \partial_i q^i - \partial_i \rho^{ji} \partial_j v - \rho^{ji} \partial_i \partial_j v \nonumber \\
										 &=   v \partial_i q^i- \rho^{ji} \partial_i \partial_j v 
\end{align}
Introducing now the closure approximation from above we arrive at the symbolic equation
\begin{equation} \label{Eq: evolution qt}
 \partial_t q_{(0)}^\mathrm{old} = - v \divergence\left(\bar{k} \kappavec^\perp \right)  -  \frac{1}{2|\kappavec|^2} \left[ \left( \rho_{0} + |\kappavec| \right) \kappavec \otimes \kappavec + \left( \rho_{0} - |\kappavec| \right) \kappavec^\perp \otimes \kappavec^\perp \right] : \nabla \nabla v,
\end{equation}
which recovers the evolution equation provided in \citet{hochrainer_etal14} in the current notation. Note that $ \nabla \nabla v $ denotes the Hessian of $ v $. The former equations therefore are very closely related to the new ones and employ the same closure approximation, which had been derived with somewhat different arguments before. The essential difference between the new Eq.\ \eqref{Eq: evol q simplified} and the former one, Eq.\ \eqref{Eq: evolution qt}, is that the new one keeps the conservative form also after making the closure approximation. By contrast, due to the separate closure assumptions for $ \q_{(1)} $ and $ \rhobf_{(2)} $ the assumed tensors are no longer related by a the divergence relation employed in calculation Eq.\ \eqref{Eq: transformation q v iso} such that the right hand side of Eq.\ \eqref{Eq: evolution qt} is in general no longer the divergence of a vector. Note that for the case of only geometrically necessary dislocations both approximations are the same.
\subsubsection{Second order closure} \label{Sec: Second order closure}
In the above subsection we closed the hierarchy of evolution equations for an isotropic velocity at the lowest possible order. The resulting theory and its older variant are already surprisingly complete and display many effects not contained in any other dislocation based plasticity theory available in the literature. Nevertheless, some details get lost in the simplified theory as compared to the higher dimensional one \citep{hochrainer_etal14, monavari_zs14}. We will now close the system at next higher order, that is, by additionally following the evolution of the second order dislocation density tensor as obtained from Eq.\ \eqref{Eq: evolution reducible pure glide iso v tensors 2}, 
\begin{equation} \label{Eq: evolution A iso v} 
	\partial_t \rho^{ij}		= \left[ - \varepsilon^{ik} \partial_k \left( v \kappa^{j} \right)  +  v Q ^{i j}  -  \tensor{\varepsilon}{^{j}_{k}} \rho^{i k m} \partial_m v \right]_{\mathrm{sym}}.
\end{equation}
This needs to be seconded by (\ref{Eq: evol kappa iso v}) for the evolution of $ \kappavec = \rhobf_{(1)} $ and (\ref{Eq: evol q iso v}) for $ q_{(0)} $. Note that $ \rho_{(0)} $ is obtained as the trace of $ \rhobf_{(2)} $ and $ \q_{(1)} $ as its divergence such that their evolution does not need to be traced separately. In this case we need closure assumptions for $ \Q_{(2)} $ and $\rhobf_{(3)}$. 

The second order curvature tensor $ \Q_{(2)} $ is a symmetric tensor given through
\begin{equation} \label{Eq: definition q_2 perp} 
	Q^{i j}		= \int_{S^1} q \tensor{\varepsilon}{^{i}_{k}} l^k \tensor{\varepsilon}{^{j}_{m}} l^m \dS.
\end{equation}
This tensor is structurally similar to the second order dislocation density tensor $ \rhobf_{(2)} $ such that we employ an analogous closure assumption as for the latter before. Note that the trace of this tensor is $ q_{(0)} $ and that the vector corresponding to $ \kappavec $ in the closure assumption Eq.\ \eqref{Eq: closure rho2} now is $ \q_{(1)} $. Accordingly we make the closure assumption
\begin{equation} \label{Eq: closure q_2 perp}
  Q^{i j} := \frac{1}{ 2q_m q^m } \left[ \left( q_{0} + \sqrt{ q_m q^m } \right) q^i q^j + \left( q_{0} - \sqrt{ q_m q^m } \right) \tensor{\varepsilon}{^{i}_{k}} q^k \tensor{\varepsilon}{^{j}_{l}} q^l \right].
\end{equation}
From the third order tensor $\rhobf_{(3)}$ we know that its trace is $ \rhobf_{(1)} = \kappavec $. On the other hand we do again know that for the case of only geometrically necessary dislocations we have
\begin{equation} \label{Eq: closure rho3}
  \rhobf_{(3)} := \frac{1}{ |\kappavec|^2} \kappavec \otimes \kappavec \otimes \kappavec,
\end{equation}
which we directly adopt as the closure assumption. The complete set of evolution equations closed at second order for an isotropic velocity is
\begin{align}  
	\partial_t \kappa^{i}	&= -\varepsilon^{i j} \partial_j \left( v \tensor{\rho}{^{k}_{k}} \right),  \label{Eq: evolution second order iso v 1} \\
	\partial_t \rho^{ij} &= \left\lbrace - \varepsilon^{ik} \partial_k \left( v \kappa^{j} \right)  + \frac{v}{ 2q_l q^l } \left[ \left( q_{0} + \sqrt{ q_l q^l } \right) q^i q^j + \left( q_{0} - \sqrt{ q_l q^l } \right) \tensor{\varepsilon}{^{i}_{k}} q^k \tensor{\varepsilon}{^{j}_{l}} q^l \right]  -  \right. \nonumber \\ &{} \qquad \qquad \qquad \qquad \qquad \qquad \qquad \qquad \qquad \qquad \qquad \qquad \qquad \qquad \left. \frac{1}{ \kappa_l \kappa^l } \tensor{\varepsilon}{^{j}_{k}} \kappa_i\kappa_k\kappa_m \partial_m v \right\rbrace_{\mathrm{sym}},  \label{Eq: evolution second order iso v 2} \\
	\partial_t q_{(0)} &=  \partial_i \left( v  q^i - \rho^{ji} \partial_j v \right), \label{Eq: evolution second order iso v 3} \\
	q^i &= \partial_j \rho^{ji}  .
\end{align}
In order to connect to equations for edge and screw dislocations (cf. \citet{arsenlis_etal04}) we now fix the coordinate system such that the 1-direction points in Burgers vector direction while the 2-direction is the positive edge direction. We then find for the components of the second order tensor
\begin{align}  
	\partial_t \rho^{11} &= \partial_2 \left( v \kappa^{1} \right)  + \frac{v}{ 2 } \left[ q_{0} + \frac{1}{\sqrt{ q_l q^l }} \left( q^1 q^1 - q^2 q^2 \right) \right]  +  \frac{ \kappa_1\kappa_2}{ \kappa_l \kappa^l } \left( \kappa_1 \partial_1 v + \kappa_2 \partial_2 v \right) ,  \label{Eq: evolution second order iso v coord 1} \\
	\partial_t \rho^{12} &= \frac{1}{2} \left[ \partial_2 \left( v \kappa^{2} \right) - \partial_1 \left( v \kappa^{1} \right) + \frac{2vq^1 q^2}{\sqrt{ q_l q^l }}  -  \frac{\kappa_1 \kappa_1 -\kappa_2 \kappa_2}{ \kappa_l \kappa^l } \left( \kappa_1 \partial_1 v + \kappa_2 \partial_2 v \right) \right], \label{Eq: evolution second order iso v coord 2} \\
	\partial_t \rho^{22} &=  - \partial_1 \left( v \kappa^{2} \right)  + \frac{v}{ 2 } \left[ q_{0} - \frac{1}{\sqrt{ q_l q^l }} \left( q^1 q^1 - q^2 q^2 \right) \right] -  \frac{ \kappa_1\kappa_2}{ \kappa_l \kappa^l } \left( \kappa_1 \partial_1 v + \kappa_2 \partial_2 v \right). \label{Eq: evolution second order iso v coord 3}   
\end{align}
Equations \eqref{Eq: evolution second order iso v coord 1}--\eqref{Eq: evolution second order iso v coord 3} answer a very important question which had only received very unsatisfying answers in earlier screw--edge based theories. Typically, the fluxes of edge and screw dislocations naturally had the first term in the diagonal elements but otherwise were not coupled aside from the other dislocation type possibly influencing the dislocation velocity. Moreover, no exchange between the types was considered and, if at all, line length changes were modeled by a rectangular dislocation model, where the motion of one type produces line length of the other, see, for example, \citet{arsenlis_etal04}. All this can be handled consistently with Eqs.\ \eqref{Eq: evolution second order iso v coord 1}--\eqref{Eq: evolution second order iso v coord 3}.

We finally discuss the two closure assumptions made in the current second order theory. The closure assumption for the second order curvature tensor, Eq.\ \eqref{Eq: closure q_2 perp}, has a spherical part, which ensures the correct line length increase (i.e., $v q_{(0)}$) represented in the total dislocation density $ \rho = \rho^{11} + \rho^{22} $ obtained as the trace of the second order tensor. The deviatoric part of the tensor enhances the line length increase in the direction of the curvature vector $ \q^{(1)} $ and decreases it in the orthogonal direction. The closure assumption for the third order tensor, Eq.\ \eqref{Eq: closure rho3}, implies a rotation of the main axes of the second order tensor proportional to the gradient of the velocity in the direction of the dislocation density vector $ \kappavec $. Note that the  deviatoric tensor multiplied with this gradient expression is given by $ \kappavec \sotimes \kappavec^\perp $ which has the eigendirections $ \kappavec + \kappavec^\perp $ and $ \kappavec - \kappavec^\perp $. We recall that in the first order closure approximation Eq.\ \eqref{Eq: closure rho2} we assumed the main axes of  the second order tensor to be aligned with $\kappavec$ and $  \kappavec^\perp $, which is exact for pure GND configurations. The eigendirections of the tensor in front of the gradient expression in Eq.\ \eqref{Eq: evolution second order iso v 2} are tilted by $45^\circ$ to these directions. Consequently, this term introduces the maximal change of main axes direction and thus a rotation.

A first proof of principal for the above equations has been provided in \citet{hochrainer14} where we simulated the expansion of a single loop in an inhomogeneous velocity field. A more detailed comparison between the higher dimensional theory, the second order theory and the lowest order closure will be provided elsewhere.
\subsubsection{Elliptic velocity} \label{Sec: Elliptic velocity}
In this subsection we present a brief outlook on what closure assumptions will be needed when considering less trivial assumptions on the angular dependence of the dislocation velocity. In many materials, as e.g.\ in bcc crystals, the mobility of screw and edge dislocations is very different. Translated to the current treatment this results in an anisotropic, that is, angular dependent dislocation velocity $ v \left( \varphi \right) $. A simple approach on modeling such anisotropy is given (cf. \citet{cai_b04}) by
\begin{equation}
	v \left( \varphi \right) = v_{\mathrm{s}} \cos^2 \varphi + v_{\mathrm{e}} \sin^2 \varphi,
\end{equation}    
where $v_{\mathrm{s}}$ and $v_{\mathrm{e}}$ are the velocities of screw and edge dislocations, respectively. This can be equivalently rewritten as
\begin{equation}
	v \left( \varphi \right) = v_{ij} l^i\left(\varphi\right) l^j\left(\varphi\right)
\end{equation}    
with $ v_{11} = v_{\mathrm{s}} $, $ v_{22} = v_{\mathrm{e}} $, and $ v_{12} = v_{21} = 0 $ in the coordinate system used in the case of glide only. In other words, this assumes a representation of the dislocation velocity by its second order tensor in the expansion discussed in Section \ref{Sec: Velocity expansion} with the main axes coinciding with edge and screw directions.

In this case we obtain the rotation velocity as
\begin{equation}
	\vartheta \left( \varphi \right) = - \nabla_{\L} v \left( \varphi \right) = - \nabla_{\l} v_{ij} l^i\left(\varphi\right) l^j\left(\varphi\right) - \frac{q}{\rho} v_{ij} \left[ \tensor{\varepsilon}{^{i}_{k}} l^k\left(\varphi\right) l^j\left(\varphi\right) + l^i\left(\varphi\right) \tensor{\varepsilon}{^{j}_{k}} l^k \left(\varphi\right) \right],
\end{equation}    
For the evolution we therefore obtain
\begin{align} \label{Eq: evolution reducible pure glide elliptic v}
	\partial_t \rho_{(0)}  &=  \partial_i \left( v_{kl} \int_{S^2}  \rho  l^k l^l \tensor{\varepsilon}{^{i}_{j}} l^j \dS \right) + v_{kl} \int_{S^1} q l^k l^l \dS  \\
	\partial_t \rho^{i_1 \cdots i_k}		&= \left[ - \varepsilon^{i_1 i} \partial_i \left( v_{kl} \int_{S^1} \rho l^k l^l l^{i_2} \cdots l^{i_k} \dS \right) + (k-1) v_{kl} \int_{S^1} q l^k l^l \tensor{\varepsilon}{^{i_1}_{i}} l^i \tensor{\varepsilon}{^{i_2}_{m}} l^m l^{i_3} \cdots l^{i_k} \dS  - \right.  \nonumber \\
	&{} \qquad \qquad \left. (k-1) \int_{S^1} \rho l^k l^l l^i l^{i_1} \tensor{\varepsilon}{^{i_2}_{m}} l^m l^{i_3} \cdots l^{i_k} \dS \; \partial_i v_{kl} - \right. \nonumber \\
	&{} \qquad \qquad \qquad \qquad \left. (k-1) v_{kl} \int_{S^1} q  \left( \tensor{\varepsilon}{^{k}_{m}} l^m l^l + l^k \tensor{\varepsilon}{^{l}_{m}} l^m  \right) l^{i_1} \tensor{\varepsilon}{^{i_2}_{n}} l^n l^{i_3} \cdots l^{i_k}  \dS \right]_{\mathrm{sym}}  \\
	\partial_t q_{(0)} &=  \partial_i \left( v_{kl} \int_{S^1} q l^k l^l \tensor{\varepsilon}{^{i}_{j}} l^j \dS - \int_{S^1} \rho l^k l^l l^j l^i \dS \; \partial_j v_{kl} \right)  \\
	\partial_t q^{i_1 \cdots i_k} &= \left[ \partial_i \left( v_{kl} \int_{S^1} q l^k l^l \tensor{\varepsilon}{^{i}_{j}} l^j \tensor{\varepsilon}{^{i_1}_{m}} l^m l^{i_2} \cdots l^{i_k} \dS  - \right. \right.  \nonumber  \\
	&{} \qquad \qquad \left. \left. \int_{S^1} \rho l^k l^l l^j l^i \tensor{\varepsilon}{^{i_1}_{m}} l^m l^{i_2} \cdots l^{i_k} \dS \; \partial_j v_{kl} - \right. \right.  \nonumber  \\
	&{} \qquad \qquad \qquad \qquad \left. \left. v_{kl} \int_{S^1} q  \left( \tensor{\varepsilon}{^{k}_{m}} l^m l^l + l^k \tensor{\varepsilon}{^{l}_{m}} l^m  \right) l^{i} \tensor{\varepsilon}{^{i_1}_{n}} l^n l^{i_2} \cdots l^{i_k}  \dS \right) \right]_{\mathrm{sym}} 
\end{align}
Because these formulae are very cumbersome in the general form we write down the evolution of the first three tensors explicitly as
\begin{align} \label{Eq: evolution reducible pure glide elliptic v 1 and 2}
	\partial_t \rho_{(0)}  &=  \partial_i \left( v_{kl} \tensor{\varepsilon}{^{i}_{j}} \rho^{klj} \right) + v_{kl} \tensor{\varepsilon}{^{k}_{m}}  \tensor{\varepsilon}{^{l}_{n}} Q^{mn} \\
	\partial_t \kappa^i		&= - \varepsilon^{i j} \partial_j \left( v_{kl} \rho^{kl} \right) \\
	\partial_t \rho^{ij}		&= \left[ - \varepsilon^{i m} \partial_m \left( v_{kl} \rho ^{klj} \right) + v_{kl} Q^{ijkl} - \rho^{klmin}\tensor{\varepsilon}{^{j}_{n}} \partial_m v_{kl} - v_{kl}\left( Q^{kjli} + Q^{ljki} \right)  \right]_{\mathrm{sym}},
\end{align}
with the curvature tensors $ \Q $ introduced in Eqs.\ \eqref{Eq: Alt Q tensors 1}--\eqref{Eq: Alt Q tensors 3}. The closure approximations which are needed to close this hierarchy at low order are not obvious and will be subject of further research. We note that a promising approach for deriving general closure approximations from a maximum entropy principle on the orientation space was recently successfully employed for the lowest order closure assumption discussed in Section \ref{Sec: Lowest order closure} \citep{monavari_zs14}.

\section{Discussion and outlook} \label{Sec: Discussion}

In the current paper we transferred the concept of alignment tensors used for example in the theory of liquid crystals \citep{hess75}, polymers \citep{kroeger1998332} and fiber reinforced composites \citep{advani_t87} to the higher dimensional continuum dislocation dynamics theory of \citet{hochrainer_zg07}. In contrast to the cited applications of alignment tensors the dislocation density is described by a vectorial object on the higher dimensional space. This requires a vectorial expansion done in vector spherical harmonics and accordingly defined basis tensors. As another distinction to most other applications of alignment tensors the dislocation case also requires the consideration of alignment tensors of odd degree. From the evolution equation of the higher dimensional theory we derived hierarchies of evolution equations for expansions in vector spherical harmonics, Fourier series (in pure glide), and various irreducible and reducible tensors. These infinite hierarchies have to be terminated at low order by means of suitable closure assumptions to arrive at manageable CDD theories. For these closure assumptions we used low order terminations of an analogous expansion of the higher dimensional dislocation velocity. Considering only the scalar term of the velocity expansion we presented closure assumptions for theories terminated at first or second order for the dislocation density. We sketched what kind of closure would be needed when considering anisotropic dislocation velocities in terms of a second order velocity tensor.

The evolution equations obtained from the multipole expansion solve the basic kinematic problem in averaging dislocation systems \citep{kroener01}. The current theory unifies earlier proposed continuum dislocation density theories. Every so derived CDD contains the pseudo-continuum theory based solely on the Kröner-Nye tensor. The simplified CDD of \citet{hochrainer_etal14} is obtained as the lowest order closure with isotropic velocity and the evolution of the second order dislocation density tensor comprises and clarifies earlier introduced models based on an edge--screw decomposition of the total dislocation density. Above all, the evolution equations for the alignment tensors provide the foundation for developing CDD theories of principally arbitrary resolution of the angular dependence of dislocation state and dislocation velocity.

However, the closure assumptions made in the current paper are rather simplistic linear combinations of low order tensors which are known to hold for systems which only contain GNDs. More sophisticated closure assumptions may, for example, be obtained from a maximum information entropy principle applied to the orientation distribution function of the dislocation density. Subject to the auxiliary condition of reproducing the known low order moments (i.e. alignment tensors) one can derive a positive function on the orientation space which maximizes the information entropy and in turn minimizes the additional amount of information introduced through the closure assumption. Such a closure has recently been applied to the lowest order CDD theory by \citet{monavari_zs14} with promising results when compared to the higher dimensional CDD. However, it is unclear whether this provides a promising route for closing systems at higher order. The reason is that the case treated by \citet{monavari_zs14} could be handled analytically but the maximum entropy approach in general leads to strongly non-linear equations which can only be solve numerically. It is also worth investigating how a maximum entropy approach can be meaningfully transfered to the vectorial case we deal with in dislocation theory.

Obviously, the purely kinematic treatment presented in the current paper does not yet really constitute a continuum theory of dislocations in the sense that also the dynamics would be derived from rigorous averaging. This has so far only been done for simplified systems of parallel straight edge dislocations by \citet{groma_cz03} where pair correlations were obtained from numerical simulations. The finding that the pair correlations are short ranged allowed for the derivation of averaged velocity laws containing information on the current dislocation state. Velocity laws largely motivated from those obtained for parallel edge dislocations were used in \citet{sandfeld_etal11} and \citet{hochrainer_etal14} within the earlier developed simplified theory.  However, obtaining the averaged dynamics of three-dimensional dislocation systems will be a huge challenge. For once, already obtaining and evaluating pair correlations is very demanding and has so far only been attempted in preliminary form \citep{csikor_etal07,deng_e07}. But aside form pair correlations further three-dimensional effects as cross slip, dislocation reactions and the formation of jogs complicate the picture in three-dimensional dislocation configurations. Aside from rather generic ideas of how to deal with these issues in a higher dimensional theory \citep{elazab00} such phenomena are as yet only incorporated in the semi-phenomenological model of \citet{devincre_hk08}. Part of the mentioned effects might be contained in the pair correlations but other phenomena lead to immobilization and the subsequent emergence of dislocation sources. Although we have a concept for including dislocation sources into the higher dimensional theory \citep{hochrainer07} and the lowest order closure \citep{sandfeld_h11} this has to be coupled to the evolving dislocation state. As for the evolution of jogs a theory of single dislocations moving through a continuously dislocated medium has been presented in \citet{hochrainer13}. As a result we found that in multiple slip we can hardly expect the majority of the dislocations to be confined to a single glide plane but rather expect them to be mostly composed of interconnects between jogs introduced both by cutting polarized stationary forests as by being cut by other moving dislocations. Nevertheless, we hope that the availability of a promising kinematic framework as presented here will stimulate new efforts for developing statistical continuum theories of plasticity. 

We conclude the paper with an historical remark. Even before suggesting to use the concept of pair correlations in dislocation theory, Ekkehard Kröner suggested the so called moment theory of dislocations \citep{kroener63}. This is also an expansion of the dislocation density into a series of tensors of increasing order. However, he himself discarded the idea in favor of the pair correlation tensors \citep{kroener69}. But apparently independent of Kr\"oner's preliminary work a moment theory for dislocations has been used as a tool for stress calculations in discrete dislocation simulations by \citet{lesar_r02} and was suggested for further use by \citet{lesar_r04}. Also this time, however, the moment approach was not worked out into an actual dislocation density theory. The relation of the current alignment tensor expansion to the moment theory is not yet sufficiently clear. The moment theory expands a local distribution of dislocation lines into geometric moments similar to the inertia matrix in dynamics. It would be very interesting to see whether this sort of spatial coarse graining may be translated into the current tensor expansion which derives rather from an ensemble averaging. If a correspondence could be obtained one may infer that the eigenstress expansion introduced in \citet{lesar_r02} might proof useful also in continuum dislocation dynamics.

\bibliographystyle{elsarticle-harv}

\end{document}